\documentclass[american,aps,pra,twocolumn,superscriptaddress]{revtex4-2}
\usepackage[colorlinks=true, allcolors=myurlcolor]{hyperref}
\usepackage{graphics,graphicx,amsthm,amsmath,amssymb,braket,color,framed}
\usepackage[T1]{fontenc}

\usepackage{tikz}
\usepackage{algorithm}
\usepackage{algorithmic}
\usepackage{enumerate}
\usepackage{multirow}
\usepackage{ulem}
\usepackage{tcolorbox}

\definecolor{myurlcolor}{rgb}{0,0,0.9}

\newcommand{\proj}[1]{| #1\rangle\!\langle #1 |}
\newcommand{\iinner}[2]{\langle #1 | #2\rangle}

\newcommand{\trace}{\operatorname{Tr}}
\newcommand{\Ptr}[2]{\trace_{#1}\Pa{#2}}
\newcommand{\Tr}[1]{\Ptr{}{#1}}
\newcommand{\Pa}[1]{\left[#1\right]}

\newcommand{\norm}[1]{\left\lVert #1 \right\rVert}

\newcommand{\ot}{\otimes}
\newcommand{\complex}{\mathbb{C}}
\newcommand{\real}{\mathbb{R}}
\newcommand{\NGE}{\operatorname{NGE}}
\newcommand{\MING}{\operatorname{MING}}
\newcommand{\CBB}{\mathcal{B}}
\newcommand{\CNN}{\mathcal{N}}
\newcommand{\bigO}{\mathcal{O}}
\providecommand{\T}{\mathsf{T}}

\theoremstyle{plain}
\newtheorem{thm}{Theorem}
\newtheorem{lem}[thm]{Lemma}
\newtheorem{prop}[thm]{Proposition}
\newtheorem{Def}[thm]{Definition}
\newtheorem{Rem}[thm]{Remark}

\newcommand*{\myproofname}{Proof}

\makeatother

\begin{document}

\title{Efficient Measurement of Bosonic Non-Gaussianity}

\author{Kaifeng Bu}
\email{bu.115@osu.edu}
\affiliation{Department of Mathematics, The Ohio State University, Columbus, Ohio 43210, USA}
\affiliation{Department of Physics, Harvard University, Cambridge, Massachusetts 02138, USA}

\author{Bikun Li}
\email{bikunli@uchicago.edu}
\affiliation{Pritzker School of Molecular Engineering, University of Chicago, Chicago, Illinois 60637, USA}

\begin{abstract}
    Non-Gaussian states are essential resources in quantum information processing. In this work, we investigate methods for quantifying bosonic non-Gaussianity in many-body systems. Building on recent theoretical insights into the self-convolution properties of bosonic pure states, we introduce non-Gaussian entropy as a new measure to characterize non-Gaussianity in bosonic pure states. We further propose a practical protocol for measuring non-Gaussian entropy using three beam splitters and four copies of the input state. In addition, we extend this framework to mixed states, providing a general approach to quantifying non-Gaussianity. Our results offer a convenient and efficient method for characterizing bosonic non-Gaussianity, paving the way for its implementation on near-term experimental platforms.
\end{abstract}

\maketitle

\section{Introduction}
Gaussian states, fundamental to continuous-variable (CV) quantum information, form a class of quantum states characterized by their elegant mathematical description via Gaussian Wigner distributions. This class—including coherent, squeezed, and thermal states—is particularly important due to the relative ease with which such states can be generated and manipulated using linear optical elements and Gaussian measurements~\cite{SethRMP12}. Their experimental accessibility makes them highly relevant for a wide range of practical applications, including quantum sensing and metrology~\cite{Caves81,Bondurant84,giovannetti2011advances,CorreaPRL15,Zhuang17,oh2019optimal}.

A crucial aspect of Gaussian states is that systems involving only Gaussian states, Gaussian unitary transformations, and Gaussian measurements can be efficiently simulated on a classical computer. This efficiency arises from the fact that such operations can be fully described using covariance matrices and mean vectors, which can be manipulated using polynomial-time classical algorithms~\cite{Bartlett02,Mari12,Veitch_2013}. Consequently, the inclusion of non-Gaussian elements, such as cubic phase gates or photon number resolving detectors, is essential for achieving universal quantum computation~\cite{Lloyd99,BartlettPRL02}. 

Besides, CV quantum systems have emerged as a promising platform for implementing quantum computation and demonstrating quantum advantage. Several sampling tasks have been proposed~\cite{Lund14,Douce17,Hamilton17,Cerf17}, including Gaussian boson sampling, a modification of the original boson sampling proposed by Aaronson and Arkhipov~\cite{aaronson2011computational}. This approach has attracted significant attention and has been realized experimentally~\cite{Pan20,Pan21,Jonathan22}. Furthermore, the ability to encode information in high-dimensional Hilbert spaces associated with continuous variables~\cite{Sivak2023,Ni2023,Rojkov2024,Brock2025}, offers opportunities for increased information capacity and improved robustness against noise. Therefore, the characterization and quantification of non-Gaussian states and channels are attracting significant attention.

To quantify the non-Gaussian nature of a quantum state or process,  the framework of resource theory has been developed~\cite{Genoni_PRA07,GenoniPRA08,Genoni10,Albarelli18,Takagi18,Albarelli18,ZhuangPRA18,SuPRA19,Chabaud20,Dias24,Calcluth24,hahn2024classical,WalschaersPRXQ24}. Two widely used approaches exist for the characterization and quantification of non-Gaussianity. The first approach is based on the non-negativity of the Wigner function, or equivalently, the positivity of the Husimi Q-function. Examples include the stellar rank~\cite{Chabaud20} and the logarithmic negativity of the Wigner function~\cite{Takagi18}. The theoretical basis for this characterization lies in Hudson's theorem~\cite{Hudson74,Soto83}, which states that a quantum state is Gaussian if and only if its Wigner function is non-negative. The second approach utilizes distance measures derived from resource theory, such as relative entropy~\cite{Genoni_PRA07,GenoniPRA08,Genoni10}. However, in the experiment, evaluating these measures against an unknown state requires the reconstruction of the quantum state, which leads to the exponential overhead with respect to the bosonic mode number.

In this work, we explore a new method for characterizing and quantifying non-Gaussianity based on quantum convolution associated with beam splitters, where state tomography is not needed. Our approach hinges on the violation of the classical entropic inequality $H(X+Y)\leq H(X)+H(Y)$ in the quantum setting, where $X+Y$ is the output signal of the additive noise channel with the input signal $X$ and noise $Y$. Leveraging this key insight, we introduce a new measure termed the ``non-Gaussian entropy,'' which we demonstrate to be an efficient measure for non-Gaussian bosonic pure states. 

Based on this concept, we also propose a practical experimental protocol for measuring this non-Gaussian entropy using three beam splitters
and four copies of the input state. Furthermore, for input states with zero mean, we demonstrate that the protocol can be simplified, requiring only two beam splitters and three copies of the input state.
We validate the performance of our proposed protocol through numerical simulations, even in the presence of noise. 
Additionally, we generalize this idea to quantify and measure non-Gaussianity in mixed states. Our results offer an efficient approach for characterizing and quantifying bosonic non-Gaussianity, with potential implications for quantum information processing.

\section{Basic Framework}
Consider a CV quantum system consisting of $N$ bosonic modes. The associated Hilbert space is $\mathcal{H} = \bigotimes_{i=1}^N \mathcal{H}_i$, with annihilation and creation operators $\{ \hat{a}_k, \hat{a}_k^\dagger \}_{k=1}^N$. These operators can be arranged into the vector operator $\hat{R} = (\hat{a}_1, \hat{a}_1^\dagger, \dots, \hat{a}_N, \hat{a}_N^\dagger)$,
which satisfies the canonical commutation relation $[\hat{R}_i, \hat{R}_j] = \Omega_{ij}, \quad \forall i,j \in \{1,2,\dots,2N\}$,
where $\Omega = \bigoplus_{i=1}^N \omega$, and $\omega = \begin{bmatrix} 0 & 1 \\ -1 & 0 \end{bmatrix}$.
The quadrature (position and momentum) operators are defined as $\hat{q}_k = (\hat{a}_k + \hat{a}_k^\dagger)/\sqrt{2}$, $\hat{p}_k = i(\hat{a}_k^\dagger - \hat{a}_k)/\sqrt{2}$, which can be grouped into the vector operator $\hat{x} = (\hat{q}_1, \hat{p}_1, \dots, \hat{q}_N, \hat{p}_N)$, satisfying $[\hat{x}_i, \hat{x}_j] = i \Omega_{ij}$. For a state $\rho$, the expectation value of an operator $\hat{O}$ is given by $\langle \hat{O} \rangle = \Tr{\rho \hat{O}}$.

Gaussian states and Gaussian unitaries are central concepts in CV quantum systems. A state is Gaussian if its Wigner function is Gaussian. A unitary is Gaussian if it is generated by a Hamiltonian that is quadratic in $\{ \hat{a}_k, \hat{a}_k^\dagger \}$. For example, the beam splitter acting on subsystems $A$ and $B$ is defined by
\begin{equation}\label{eq:U_BS}
U_{\theta} := \exp\left[\theta \sum_{i=1}^N \left( \hat{a}_{i,A}^\dagger \hat{a}_{i,B} - \hat{a}_{i,A} \hat{a}_{i,B}^\dagger \right)\right],
\end{equation}
where $\hat{a}_{i,A}$ and $\hat{a}_{i,B}$ are the annihilation operators for the $i$th mode in subsystems $A$ and $B$, respectively. The beam splitter is called 50:50 when $\theta = \pi/4$, performing addition and subtraction on the input quadratures (up to a scaling factor of $1/\sqrt{2}$).
This motivates the definition of quantum convolution~\cite{Cushen71,Datta21,beigi2023towards,BGJ23a}:
\begin{equation}\label{eq:QConv}
    \begin{aligned}
        \rho\boxplus\sigma &:=\Ptr{B}{U_{\pi/4}(\rho\otimes\sigma) U^\dag_{\pi/4}},\\
        \rho\boxminus\sigma &:=\Ptr{A}{U_{\pi/4}(\rho\otimes\sigma) U^\dag_{\pi/4}},
    \end{aligned}
\end{equation}
where $\rho$ and $\sigma$ are bosonic states defined on subsystems $A$ and $B$, respectively.

\begin{figure}
    \centering
    \includegraphics[width=0.99 \linewidth]{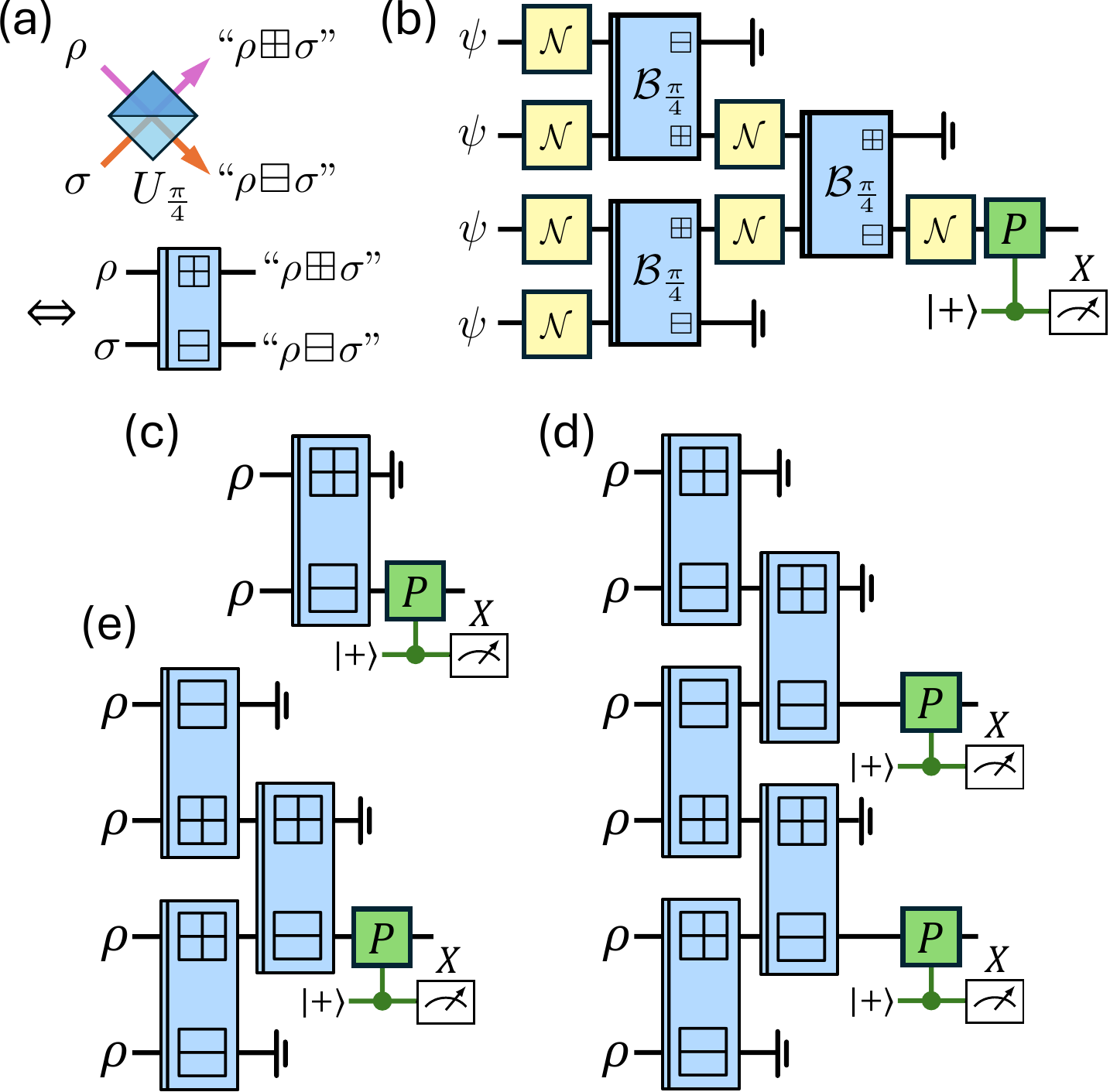}
    \caption{In this figure, $\psi$ denotes a pure bosonic state, while $\rho$ and $\sigma$ represent general bosonic states. The green box with green lines represents the controlled-$P$ gate. (a) A 50:50 beam splitter entangles the input state $\rho \otimes \sigma$, such that the \textit{reduced} output states correspond to $\rho \boxplus \sigma$ and $\rho \boxminus \sigma$. The non-Gaussianity of $\rho$ is inferred from the correlations between these two outputs. The quantum gate $U_{\pi/4}$ is depicted as a blue box, with $\boxplus$ and $\boxminus$ indicating the output directions. (b) The circuit for measuring $\NGE_{2,1}(\psi)$ of an unknown pure state $\psi$ is based on a beam splitter and parity measurements. It represents a noisy version of panel (e), in which a noisy beam splitter $\CBB_{\pi/4}$ and a noisy channel $\CNN$ are applied. Panels (c), (d), and (e) show swap-test-based circuits designed to estimate the non-Gaussianity $d_F(\rho)$ of a general, unknown mixed state $\rho$.
    }
    \label{fig:BS_protocol}
\end{figure}

\section{Main results}

\subsection{Pure states}
We begin with a classical model where two signals, $X$ and $Y$, are added. The output $X + Y$ corresponds to the classical convolution of $X$ and $Y$. Let $H(\cdot)$ denote the Shannon entropy of a random variable. Then, the convolution satisfies the fundamental subadditivity inequality:
\begin{align}\label{ineq:CConv}
    H(X+Y)\leq H(X)+H(Y),
\end{align}
If $Y$ is an independent and identically distributed copy of $X$, this simplifies to $H(X + X) \leq 2H(X)$.

In CV quantum systems, the analogs of $X$ and $Y$ are quadrature operators $\hat{x}$ and $\hat{y}$. As discussed earlier, the quantum analog of convolution is implemented by a beam splitter acting on two states (Eq.~\eqref{eq:QConv}). In contrast to the classical case, we find that non-Gaussian states can violate the subadditivity inequality in Eq.~\eqref{ineq:CConv} under quantum convolution. This violation is a key observation that motivates our test and measure of non-Gaussianity.

\begin{thm}[Violation of the Subadditive Inequality for Convolution]
\label{thm:pure}
Given a pure bosonic state $\psi$, the following subadditive inequality:
\begin{align}\label{ineq:QCONV}
S(\psi \boxplus \psi) \leq 2S(\psi),
\end{align}
holds if and only if $\psi$ is a Gaussian state, where $S(\rho)=-\Tr{\rho\log_2 \rho}$ is the
von Neumann entropy (base 2). Hence, the subadditive inequality \eqref{ineq:QCONV} is violated if and only if $\psi$ is a non-Gaussian state.
\end{thm}

The detailed proof of the above theorem and other results for general pure states in this section are presented in Appendix~\ref{appen:A}.
It is worth noting that the results discussed above also apply to stabilizer states in  qudit/qubit systems~\cite{BGJ24a, BGJ23a,BGJ23b,BGJ23c}, as well as fermionic Gaussian states in fermionic systems~\cite{lyu2024fermionic,lyu2024fermionic_G} with different definitions of quantum convolutions.

\begin{figure}
    \centering
    \includegraphics[width=0.99 \linewidth]{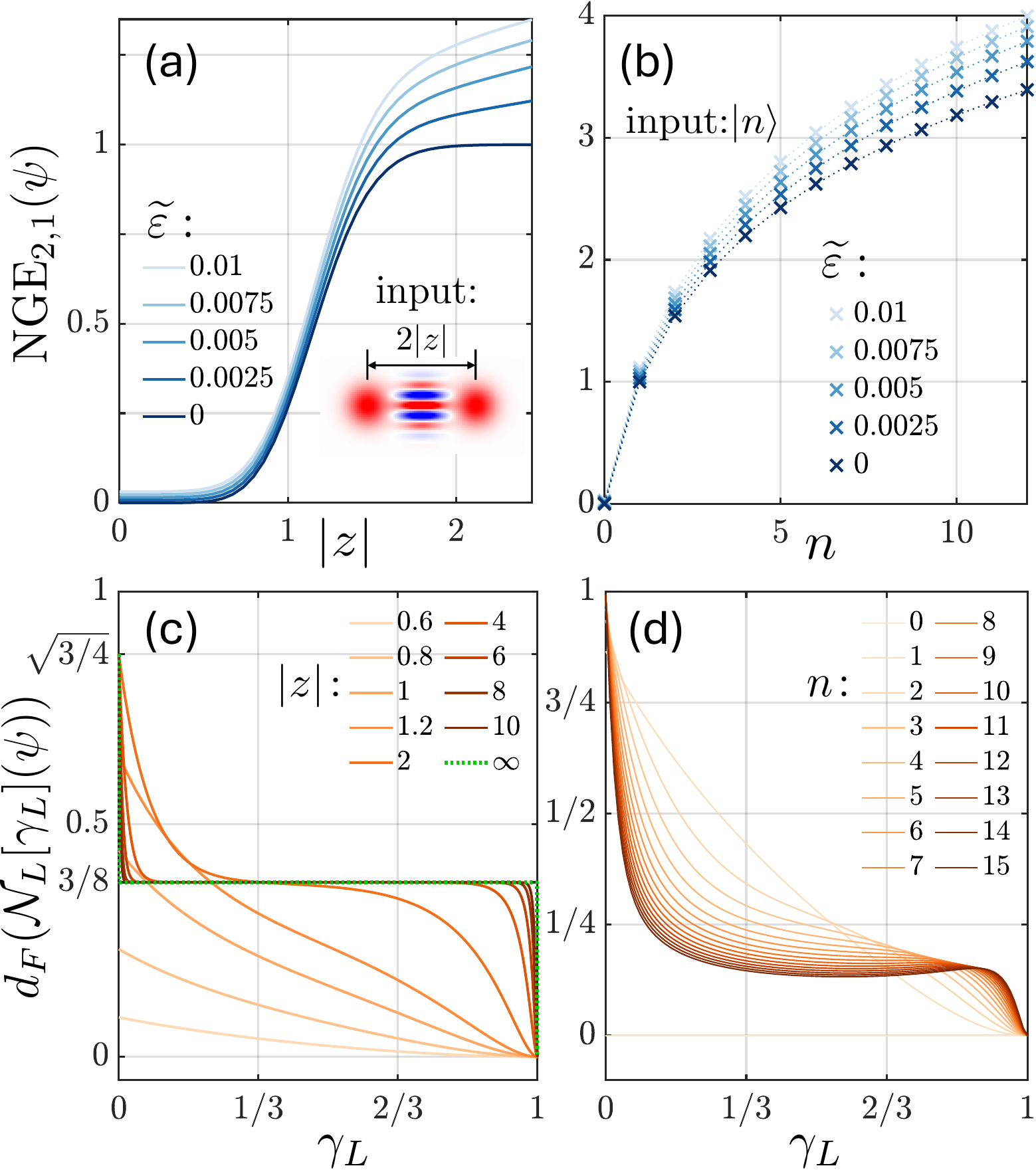}
    \caption{Panels (a) and (b) show simulations of $\NGE_{2,1}$ measured using the noisy circuits from Fig.~\ref{fig:BS_protocol}(b), with a pure input state $\psi$. Panels (c) and (d) display the exact values of $d_F$, measured using noise-free circuits, where the input state is the mixed state $\CNN_L[\gamma_L](\psi)$. The input $\psi$ is chosen as a two-component cat state $\propto \ket{z} + \ket{-z}$ (inset) in panels (a) and (c), and as a Fock state $\psi = \proj{n}$ with particle number $n$ in panels (b) and (d).
    }
    \label{fig:curves}
\end{figure}

The above theorem shows that $\psi \boxplus \psi$ is a pure state if and only if $\psi$ is Gaussian. This equivalence reduces the task of determining whether a pure state $\psi$ is Gaussian to estimating the purity of the beam splitter output $\psi \boxplus \psi$. 
It is well known that the purity of quantum states can be estimated via swap tests. Remarkably, by exploiting quantum interference at a beam splitter, the swap test can be implemented without complex circuitry. The method involves interfering two identical copies of the state using an additional $50:50$ beam splitter, followed by a parity measurement on one of the output modes. 
This parity measurement directly reflects the overlap between the two input states, thereby enabling an estimate of the purity~\cite{DaleyPRL12,AlvesPRL04,EkertPRL02,BovinoPRL05,Islam15}.
Formally, let the number operator be $\hat{n}_{i,B}:=\hat{a}_{i,B}^\dag \hat{a}_{i,B}$, and define the total parity operator as $P:=(-1)^{\sum_{i=1}^N\hat{n}_{i,B}}$. Then the average parity is related to the overlap of states by the following equation:
\begin{equation}\label{eq:innerprod_parity}
     \Tr{\rho\sigma} = \Tr{(\rho\boxminus \sigma)P}\;.
\end{equation}
This arises because a 50:50 beam splitter transforms the annihilation operators of two bosonic modes in subsystems $A$ and $B$ as follows: $\hat{a}_{i,A} \to (\hat{a}_{i,A} + \hat{a}_{i,B})/\sqrt{2}$, and $\hat{a}_{i,B} \to (\hat{a}_{i,B} - \hat{a}_{i,A})/\sqrt{2}$. Under this transformation, the symmetric subspace of the two modes $\hat{a}_{i,A}$ and $\hat{a}_{i,B}$ is mapped to the subspace of states with an even number of particles in mode $\hat{a}_{i,B}$, while the antisymmetric subspace is mapped to states with an odd number of particles in the same mode. Thus, measuring the parity in mode $\hat{a}_{i,B}$ effectively distinguishes between the symmetric and antisymmetric components, enabling direct estimation of the purity $\Tr{\rho^2}$ or the inner product $\Tr{\rho\sigma}$, depending on the input states.

Based on the above observation, we propose a method to test and measure non-Gaussianity that leverages quantum convolution and purity measurements implemented with beam splitters. This protocol efficiently utilizes only four copies of the input states, as detailed in the following steps.
The goal of this protocol is to obtain the parity of $(\psi\boxplus \psi)\boxminus(\psi\boxplus \psi)$, which is equivalent to the purity of $\psi\boxplus \psi$: 
\begin{align}\label{eq:aver_par}
\langle P\rangle
=\Tr{(\psi\boxplus \psi)^2}.
\end{align}
Hence, based on the above analysis, we have the following results regarding the parity. 

\begin{center}
\begin{tcolorbox}[width=8cm,height=3.5cm,title=Bosonic Gaussian Test]\label{prot:BG}
(1) Prepare 4 copies of the pure states $\psi$ and perform the 50:50 beam splitter to get 2 copies of  $\psi\boxplus\psi$.

(2) Perform the 50:50 beam splitter on 2 copies of  $\psi\boxplus\psi$, and measure the parity on the second half of the system. (See Fig.~\ref{fig:BS_protocol})
\end{tcolorbox}
\end{center}

\begin{thm}\label{thm:Gauss_Test}
The average parity $\langle P\rangle=1$ iff $\psi$ is a Gaussian state. In addition, 
if the maximal overlap with the Gaussian state is $\max_{\phi_G}|\iinner{\psi}{\phi_G}|^2=1-\epsilon$, then $\langle P\rangle\geq (1-\epsilon)^2$.
\end{thm}

We can generalize the above idea to introduce some measures to quantify the bosonic non-Gaussianity in terms of R\'enyi entropies and higher-order quantum convolutions.
\begin{Def}
The $(\alpha,k)$ R\'enyi non-Gaussian entropy for any pure state $\ket{\psi}$ is
\begin{align}
    \NGE_{\alpha,k}(\psi)=S_{\alpha}(\boxplus^k\psi),
\end{align}
where $S_{\alpha}(\rho)=\frac{1}{1-\alpha}\log_2 \Tr{\rho^\alpha}$ is the  quantum R\'enyi entropy, and the higher-order quantum convolution is defined inductively as $\boxplus^0\psi=\psi$, and $\boxplus^{k}\psi=(\boxplus^{k-1}\psi)\boxplus\psi$ for any integer $k\geq 1$.
\end{Def}
By the definition, the average parity $\langle P\rangle$ in \eqref{eq:aver_par} is equivalent to $\NGE_{2,1}(\psi)$ via the relation $\langle P\rangle=2^{-\NGE_{2,1}(\psi)}$. Hence, the bosonic Gaussian test (as shown in Fig.~\ref{fig:BS_protocol}) also provides an experimental method to measure $\NGE_{2,1}(\psi)$.

\begin{prop}\label{prop:non_entrop}
The $(\alpha, k)$ R\'enyi non-Gaussian entropy $\NGE_{\alpha,k}(\psi)$ with $\alpha\in [0,+\infty]$ and $k\in \mathbb{Z}_+$ satisfies  the following properties:
 \begin{enumerate}[(1)]
    \item Faithfulness: $\NGE_{\alpha,k}(\psi)\geq 0$ with equality iff $\psi$ is a Gaussian state. 
    \item Gaussian invariance: $\NGE_{\alpha,k}(U_G\proj{\psi} U^\dag_G)=\NGE_{\alpha,k}(\psi)$ for any Gaussian unitary $U_G$. 
    \item Additivity under tensor product: $\NGE_{\alpha,k}(\psi_1\ot\psi_2)=\NGE_{\alpha,k}(\psi_1)+\NGE_{\alpha,k}(\psi_2)$. 
\end{enumerate}    
\end{prop}
These properties guarantee that the   $(\alpha, k)$ R\'enyi non-Gaussian entropies can be used as  measures of non-Gaussianity for pure bosonic states.

If the quadratures of the input pure state $\psi$ have zero mean, then the protocol to test Gaussianity can be simplified by using three copies of the input states based on the following statement: A pure state $\psi$ with zero mean is Gaussian iff $\psi \boxplus \psi=\psi$. The proof of this statement, and the corresponding protocol to test Gaussianity for zero-mean pure states, are detailed in Appendix~\ref{appen:zero_mean}.

The protocol for measuring $\NGE_{2,1}$ is based on the swap test~\cite{barenco1997stabilization,BuhrmanPRL01,de2019quantum}. For an $N$-mode bosonic input state, estimating the average parity $\langle P \rangle$ within additive error $\varepsilon$ requires $\bigO(1/\varepsilon^2)$ copies of the input state. Assuming a two-mode beam splitter is used as a quantum gate, each run of the protocol involves $\bigO(N)$ gates, while the circuit depth remains $\bigO(1)$, since all beam splitters act on disjoint mode pairs and can be executed in parallel. More generally, the time complexity for measuring $\NGE_{2,k}$ scales as $\bigO(\log k)$ when parallelization is employed.
Our approach offers a significant advantage over methods requiring quantum state tomography, which demand $\bigO(e^{\kappa n}/\varepsilon'^2)$ copies of the state to reconstruct it up to trace distance error $\varepsilon'$, where $\kappa>0$ is a constant dependent on the Hilbert space dimension truncation.

\subsection{Mixed states}
We have studied the measurement of bosonic non-Gaussianity for pure states. However, in real experiments, the quantum state $\rho$ is generally mixed. We now extend our method to this general case with the following theorem.

\begin{thm}\label{thm:mix}
A bosonic state $\rho$ is Gaussian iff $U_{\pi/4}(\rho\otimes \rho) U^\dag_{\pi/4}$ is a product state.
\end{thm}

For simplicity, we define $\rho_{AB} := U_{\pi/4} (\rho \otimes \rho) U^\dagger_{\pi/4}$, $\rho_A := \rho \boxplus \rho$, and $\rho_B := \rho \boxminus \rho$. The theorem above indicates that the deviation between $\rho_{AB}$ and $\rho_A \otimes \rho_B$ can serve as a measure of non-Gaussianity. Accordingly, a natural generalization of $\NGE_{\alpha,k}$ is the relative entropy $D_\alpha(\rho_{AB} \| \rho_A \otimes \rho_B)$, which functions well as a non-Gaussianity monotone (see Appendix~\ref{appen:B} for examples involving the $\alpha$-mutual information of non-Gaussianity).
However, to the best of the authors’ knowledge, relative entropy and its generalizations do not lend themselves to convenient measurement without the usage of full quantum state tomography. This is because relative entropy is typically defined in terms of non-analytic functions of density matrices, making it unsuitable for direct evaluation using a finite set of swap-test-like circuits without approximation.
To address this, we propose a non-Gaussianity measure, termed the Frobenius measure, which enables efficient measurement:
\begin{equation}
    d_F(\rho) := \|\rho_{AB} - \rho_A\otimes \rho_B\|_F \;,
\end{equation}
where $\|M\|_F := \sqrt{\Tr{M^\dagger M}}$ denotes the Frobenius norm. An advantage of this measure is that $d_F(\rho)^2$ decomposes into a sum of three terms: $\Tr{\rho_{AB}^2}$, $\Tr{\rho_A^2} \Tr{\rho_B^2}$, and $-2\Tr{\rho_{AB}(\rho_A \otimes \rho_B)}$, which can be measured using the swap-test-based circuits shown in Fig.~\ref{fig:BS_protocol}(c), (e), and (d), respectively. In particular, the last term is obtained from the expectation value of the product of two parity measurements in Fig.~\ref{fig:BS_protocol}(d). In other words, the previous analysis regarding protocol complexity still holds. Example values of $d_F$ are presented in Fig.~\ref{fig:curves}(c,d).

\begin{prop}\label{thm:main_mix}
The Frobenius measure $d_F(\rho)$ satisfies the following properties: 
\begin{enumerate}[(1)]
    \item Faithfulness: $d_F(\rho)\geq 0$ with equality iff $\rho$ is a Gaussian state. 
    \item Gaussian invariance: $d_F(U_G\rho U^\dag_G) = d_F(\rho)$ 
    for any Gaussian unitary $U_G$.
    \item Subadditivity under tensor product: $d_F(\rho\otimes \sigma)\leq d_F(\rho)+ d_F(\sigma)$. 
\end{enumerate}
\end{prop}
The proof of the above proposition is given in Appendix~\ref{appen:B}. This result establishes that the Frobenius measure $d_F$ is a valid and experimentally accessible measure of non-Gaussianity. However, we emphasize that $d_F$ is not a non-Gaussianity monotone: the inequality $d_F(\rho) \ge d_F(\Lambda_G(\rho))$ does not generally hold for Gaussian operations $\Lambda_G$, as the Frobenius norm is not contractive under quantum channels~\cite{perez2006contractivity}. For instance, Fig.~\ref{fig:curves}(d) shows that for Fock states of large $n$, particle loss can increase $d_F$, demonstrating its non-monotonic behavior.

\subsection{Numerical results}\label{sec:numerical}
In practical experiments, the accuracy of NGE measurements is affected by noise in the quantum circuit. As shown in Fig.~\ref{fig:BS_protocol}(b), each bosonic mode in the rotating frame experiences a noise channel $\CNN$, defined as
\begin{equation}\label{eq:N_chn}
\CNN[\sigma_D^2,\sigma_P^2,\gamma_L] := \CNN_D[\sigma_D^2] \circ \CNN_P[\sigma_P^2] \circ \CNN_L[\gamma_L],
\end{equation}
where $\CNN_L[\gamma_L]$ is a bosonic loss channel with loss rate $\gamma_L$, $\CNN_P[\sigma_P^2]$ is a dephasing channel with dephasing variance $\sigma_P^2$, and $\CNN_D[\sigma_D^2]$ is a random displacement channel with displacement variance $\sigma_D^2$. Due to the commutativity of these channels, any ordering of their composition can be rewritten in the form of Eq.~\eqref{eq:N_chn} with appropriate parameters. Moreover, $\CNN$ can be generated by a master equation with suitable jump operators (see Appendix~\ref{appen:exam}).
To model a noisy beam splitter, we replace $U_{\pi/4}$ with a random unitary channel $\CBB_{\pi/4} \sim \{\sqrt{p_\varphi} U_{\pi/4+\varphi} \}_\varphi$, where $\varphi$ is sampled from a zero-mean Gaussian distribution with variance $\sigma_B^2$. For simplicity, we simulate weak noise with $\sigma_D = \sigma_P = \gamma_L = \sigma_B = \widetilde{\varepsilon} \ll 1$.
The final parity measurement is performed using an ancilla qubit initialized in the $\ket{+} \propto \ket{0} + \ket{1}$ state~\cite{footnote1}. A controlled-parity gate can be implemented via the interaction $V \sim \sum_{j=1}^N \hat{n}_{B,j} (\proj{0} - \proj{1})$~\cite{Sun2014}. The average parity is then extracted from the expectation value of the Pauli-$X$ operator: $\langle P \rangle = \langle X \rangle$.
In the practical experiment, due to imperfect preparation, gate implementation, or readout, the phase-flip errors in the ancilla qubit can introduce a rescaling: $\langle P \rangle \to \langle P \rangle (1 - 2\varepsilon_p)$, where $\varepsilon_p$ is the error probability. We include this effect by setting $\varepsilon_p = \widetilde{\varepsilon}$.
As two examples, we consider the cases where the input state $\psi$ are the two-component cat state $\propto \ket{z} + \ket{-z}$ (with $\ket{z} := e^{z\hat{a}^\dagger - z^* \hat{a}} \ket{0}$) and the Fock state $\proj{n}$. The simulation results are presented in Fig.~\ref{fig:curves} (a,b).

\section{Conclusion and future work}
In this work, we propose a novel and efficient method for quantifying bosonic non-Gaussianity, based on quantum convolution as implemented via beam splitters. A key advantage of our protocol is its experimental practicality: it requires only a constant number of copies of the input state, avoiding the resource-intensive demands of conventional characterization techniques. This approach provides a powerful tool for analyzing and manipulating non-Gaussian quantum states, which are essential for advanced quantum technologies such as universal quantum computation and quantum metrology.
Beyond the scope of this work, several related questions remain open. One particularly intriguing direction concerns the quantum capacity of a beam splitter acting as a noisy channel. While this question has drawn considerable attention~\cite{HolevoPRA01,caruso2006one,WolfPRL07,WildePRA12,pirandola2017fundamental,SabapathyPRA17,WildeIEEE18,Rosati2018narrow,JiangIEEE19,Jiang2020enhanced,wang2024passive,LamiPRL20,OskoueiIEEE22}, it remains unresolved whether the channel’s quantum capacity can be bounded by the non-Gaussianity of its environmental states—analogous to known results in the qudit setting~\cite{BJ24a}.
In addition, within the convex resource theory of non-Gaussianity~\cite{Takagi18,ZhuangPRA18}, free states are defined as convex mixtures of Gaussian states, which may themselves have non-Gaussian Wigner functions. Extending our method to characterize such mixtures would be a valuable direction for future research.

\section{Acknowledgements}
We thank Arthur Jaffe,  Liang Jiang and Seth Lloyd for discussions. This work was supported
in part by the ARO Grant W911NF-19-1-0302 and the ARO
MURI Grant W911NF-20-1-0082

\bibliography{reference}{}

\clearpage
\newpage
\onecolumngrid
\begin{appendix}

\section{Proof of the main results for pure states}\label{appen:A}

Consider a system of $N$ bosonic modes with Hilbert space $\mathcal{H} = \bigotimes_{i=1}^N \mathcal{H}_i$, and define the vector of quadrature operators as $\hat{x} = (\hat{q}_1, \hat{p}_1, \dots, \hat{q}_N, \hat{p}_N)$.
These operators satisfy the canonical commutation relations:
$$
[\hat{x}_i, \hat{x}_j] = i \Omega_{ij},
$$
where $\Omega = \bigoplus_{i=1}^N \omega$, and
$\omega =
\begin{bmatrix}
0 & 1 \\
-1 & 0
\end{bmatrix}$.
For a quantum state $\rho$, the expectation value of $\hat{x}_j$ is given by
$\langle \hat{x}_j \rangle = \mathrm{Tr}(\rho \hat{x}_j)$,
and the mean (or first-moment) vector is defined as
\begin{equation}
    \bar{x} = \left( \langle \hat{x}_1 \rangle, \langle \hat{x}_2 \rangle, \dots, \langle \hat{x}_{2N} \rangle \right).
\end{equation}
The covariance matrix (or second-moment matrix) $V$ is defined by
\begin{equation}
    V_{ij} = \frac{1}{2}  \mathrm{Tr} \left[ \rho \left\{ \hat{x}_i - \langle \hat{x}_i \rangle, \hat{x}_j - \langle \hat{x}_j \rangle \right\} \right].
\end{equation}

In CV systems, an important class of operators—known as displacement operators—is defined by
\begin{align}
    D(\xi)=\exp(i\hat{x}^\T\Omega\xi), 
\end{align}
where $\xi\in\real^{2N}$. The characteristic function associated with a quantum state $\rho$ is given by
\begin{align}
    \Xi_{\rho}(\xi)=\Tr{\rho D(\xi)},
\end{align}
which provides a complete description of the state $\rho$.
The Wigner function is then defined as the Fourier transform of the characteristic function:
\begin{align}
    W_{\rho}(x):=\int_{\real^{2N}}\frac{\mathrm{d}^{2N}\xi}{(2\pi)^{2N}}
    \exp(-ix^\T\Omega\xi)\Xi_{\rho}(\xi),
\end{align}
which serves as a quasi-probability distribution on the phase space.

A Gaussian state is defined as a bosonic state whose characteristic function $\Xi_\rho$ and Wigner function $W_\rho$ take Gaussian forms:
\begin{align}
\Xi_{\rho}(\xi)=&\exp\left(-\frac{1}{2}\xi^\T(\Omega \Gamma \Omega^\T)\xi-i(\Omega\bar{x})^\T\xi\right),\\
W_{\rho}(x)=&\frac{\exp\left(-\frac{1}{2}(x-\bar{x})^\T V^{-1}(x-\bar{x})\right)}{(2\pi)^N\sqrt{\det V}}.
\end{align}
Hence, by the definition, 
Thus, a Gaussian state $\rho_G$ is fully specified by its first moment vector $\bar{x}$ and covariance matrix $V$.
Furthermore, a Gaussian unitary $U_G$ is characterized by its action on the first and second moments:
\begin{align}\label{eq:UG}
\bar{x} \mapsto S \bar{x} + d, \quad V \mapsto S V S^\T,
\end{align}
where $S$ is a real symplectic matrix satisfying $S \Omega S^\T = \Omega$.

Now, let us provide another characterization of the Gaussian states based on their behavior under the beam splitter.
\begin{thm}[Restatement of Theorem \ref{thm:mix}]\label{thm:key}
For any $N$-mode bosonic state $\rho$, 
$S(\rho\boxplus\rho)+S(\rho\boxminus\rho)=2S(\rho)$
iff $\rho$ is a Gaussian state. Equivalently, 
    $\rho$ is a Gaussian state iff $U_{\pi/4}(\rho\otimes \rho) U^\dag_{\pi/4}$ is a product state.
\end{thm}
\begin{proof}
Without loss of generality, we assume that the state $\rho$ has vanishing first moments, i.e., $\bar{x} = 0$, since local displacements can always be applied to transform any state into a zero-mean form. The condition
$S(\rho\boxplus\rho)+S(\rho\boxminus\rho)=2S(\rho)$ implies that the output of the beam splitter unitary $U_{\pi/4}$ acting on $\rho \otimes \rho$ is a product state of its marginals:
\begin{align}\label{eq:implication}
     U_{\pi/4}(\rho\otimes \rho) U^\dag_{\pi/4}=
 \sigma_1\otimes \sigma_2,
\end{align}
where $\sigma_1 = \rho \boxplus \rho$ and $\sigma_2 = \rho \boxminus \rho$. As a result, the state $\rho$ can be recovered by tracing out the second subsystem after applying the inverse beam splitter operation:
\begin{align}
    \rho=\Ptr{B}{ U^\dag_{\pi/4} (\sigma_1\otimes \sigma_2) U_{\pi/4}}.
\end{align}
Consequently, the characteristic function of $\rho$ satisfies the following identity:
\begin{align*}
    \Xi_{\rho}(\xi)
    =\Xi_{\sigma_1}\left(\frac{1}{\sqrt{2}}\xi\right)\Xi_{\sigma_2}\left(-\frac{1}{\sqrt{2}}\xi\right)
    =\Xi_{\rho\boxplus\rho}\left(\frac{1}{\sqrt{2}}\xi\right)\Xi_{\rho\boxminus\rho}\left(-\frac{1}{\sqrt{2}}\xi\right)
    =\Xi_{\rho}\left(\frac{1}{2}\xi\right)^3\Xi_{\rho}\left(-\frac{1}{2}\xi\right),
\end{align*}
where the final equality follows from the identities:
\begin{align*}
    \Xi_{\rho\boxplus\rho}\left(\xi\right)
    =\Xi_{\rho}\left(\frac{1}{\sqrt{2}}\xi\right)\Xi_{\rho}\left(\frac{1}{\sqrt{2}}\xi\right),\quad
    \Xi_{\rho\boxminus\rho}\left(\xi\right)
    =\Xi_{\rho}\left(\frac{1}{\sqrt{2}}\xi\right)\Xi_{\rho}\left(-\frac{1}{\sqrt{2}}\xi\right).
\end{align*}
By iterating this procedure, we obtain the general relation:
\begin{align*}
     \Xi_{\rho}(\xi)
     =\Xi_{\rho}\left(\frac{1}{\sqrt{m_k}}\xi\right)^{a_k}
     \Xi_{\rho}\left(-\frac{1}{\sqrt{m_k}}\xi\right)^{b_k}, 
\end{align*}
where $m_k=4^k$, $a_k=\frac{4^k+2^k}{2}$, and $b_k=\frac{4^k-2^k}{2}$ for any positive integer $k$.

Following the argument in Refs.~\cite{Cushen71,WolfPRL06}, we define a classical function  $g:\real\to \complex$
by $g(x)=\Xi_{\rho}(x\xi)$.
This function is a classical characteristic function, that is, the Fourier transform of a classical probability distribution with second moment $\xi^\T \Gamma \xi / 2$. Such characteristic functions are continuous at the origin and satisfy $g(0) = 0$ and $|g(x)| \leq 1$. When the second moment is finite, $g(x)$ admits a second-order expansion:
\begin{align}
 g(x)=1-\frac{\xi^\T\Gamma\xi}{4}x^2+o(x^2).
\end{align}
It follows that
\begin{align}
 \lim_{k\to +\infty}
 g\left(\frac{x}{\sqrt{m_k}}\right)^{a_k}
 g\left(-\frac{x}{\sqrt{m_k}}\right)^{b_k}
 =    \lim_{k\to +\infty}
 \left(1-\frac{\xi^\T\Gamma\xi x^2}{4m_k}\right)^{m_k}
 =\exp\left[-\frac{1}{4}\xi^\T\Gamma\xi x^2\right].
\end{align}
Hence, 
$\Xi_{\rho}(\xi)$ is a Gaussian function, and thus $\rho$ is a Gaussian state. 
 
\end{proof}

\begin{thm}[Restatement of Theorem \ref{thm:pure}]
Given an $N$-mode pure state $\psi$, the following subadditive inequality:
\begin{align}
S(\psi \boxplus \psi) \leq 2S(\psi),
\end{align}
holds if and only if $\psi$ is a Gaussian state. 
\end{thm}
\begin{proof}
By Theorem~\ref{thm:key}, it suffices to show that $S(\psi \boxplus \psi) \leq 2S(\psi)$ iff $S(\psi\boxplus\psi)+S(\psi\boxminus\psi)=2S(\psi)$. 
First, assume that
$S(\psi\boxplus\psi)+S(\psi\boxminus\psi)=2S(\psi)=0$, then $S(\psi\boxplus\psi)=0$, and thus 
$S(\psi \boxplus \psi) \leq 2S(\psi)$.

Conversely, suppose  $S(\psi \boxplus \psi) \leq 2S(\psi)=0$, then $\psi \boxplus \psi$ is a pure state, 
and thus $  U\psi\ot\psi U^\dag$ is a product state, and thus $\psi \boxminus \psi$ is also a pure state. As a result, $S(\psi\boxplus\psi)+S(\psi\boxminus\psi)=0=2S(\psi)$.

\end{proof}

\begin{thm}[Restatement of Theorem~\ref{thm:Gauss_Test}]
    The average parity $\langle P\rangle=1$ iff $\psi$ is a Gaussian state. In addition, if the maximal overlap with the Gaussian state is $\max_{\phi_G}|\iinner{\psi}{\phi_G}|^2=1-\epsilon$, then $\langle P\rangle\geq (1-\epsilon)^2$.
\end{thm}
\begin{proof}
The first part comes directly from the above theorem. 
If the maximal overlap with the Gaussian state is 
$\max_{\phi_G}|\iinner{\psi}{\phi_G}|^2=1-\epsilon$, then 
there exists 
a pure Gaussian state $\ket{\phi_G}$ such that $|\iinner{\psi}{\phi_G}|^2=1-\epsilon$, 
i.e., the fidelity $F(\psi, \phi_G)^2=1-\epsilon$. 
Hence
\begin{eqnarray}
     F(\psi\ot\psi, \phi_G\otimes \phi_G)^2
     =(1-\epsilon)^2.
\end{eqnarray}
By the monotonicity under partial trace, 
we have 
\begin{eqnarray}\label{eq:tech}
     F(\psi\boxplus\psi, \phi_G\boxplus\phi_G)^2
     \geq (1-\epsilon)^2.
\end{eqnarray}
As $\phi_G$ is a pure Gaussian state, the output state of the quantum convolution $\phi_G\boxplus\phi_G$ is also a pure Gaussian state, i.e, $\phi_G\boxplus\phi_G=\phi'_G$.
Hence, \eqref{eq:tech} can be rewritten as 
\begin{align}
    \bra{\phi'_G}\psi\boxplus\psi\ket{\phi'_G}\geq (1-\epsilon)^2.
\end{align}
Therefore, 
    \begin{align}
    \Tr{(\psi\boxplus \psi)^2}
    \geq     \Tr{\psi\boxplus \psi \proj{\phi'_G}\psi\boxplus\psi \proj{\phi'_G}}
    \geq \bra{\phi'_G}\psi\boxplus\psi\ket{\phi'_G}^2
    \geq (1-\epsilon)^2.
\end{align}
\end{proof}

\begin{lem}\label{lem:GU_Com}
    Given a Gaussian unitary $U_G$, there exist Gaussian unitaries $U_1$  and $U_2$ that satisfy the following commutation relation:
    \begin{align}
        U_{\pi/4} (U_G\otimes U_G)
        =(U_1\otimes U_2) U_{\pi/4},
    \end{align}
    where $U_{\pi/4}$ is the $50:50$ beam splitter.
\end{lem}
\begin{proof}
Since the Gaussian unitary is fully characterized by its transformation of the first two moments, we only compare the actions of the left-hand side and right-hand side on these moments.

Let us assume that the given Gaussian unitary $U_G$
is described in \eqref{eq:UG}. The action of the left-hand side on the first moments of two input modes ($\bar{x}_A$ and $\bar{x}_B$) will be
\begin{equation}
    \left[
    \begin{array}{c}
         \bar{x}_A  \\
          \bar{x}_B
    \end{array}
    \right]
    \xrightarrow{U_G\otimes U_G}
    \left[
    \begin{array}{c}
        S(\bar{x}_A)+d  \\
        S(\bar{x}_B)+d
    \end{array}
    \right]
    \xrightarrow{U_{\pi/4}}
    \left[
    \begin{array}{c}
         S\left(\frac{\bar{x}_A +\bar{x}_B}{\sqrt{2}}\right)+\sqrt{2}d \\
          S\left(\frac{\bar{x}_B -\bar{x}_A}{\sqrt{2}}\right)
    \end{array}
    \right],    
\end{equation}
and 
\begin{equation}
\left[
    \begin{array}{cc}
      V_{AA}   & V_{AB}  \\
      V_{BA}   & V_{BB}
    \end{array}
\right]
\to 
\left[
    \begin{array}{cc}
      SV_{AA}S^\T   & SV_{AB}S^\T  \\
      SV_{BA} S^\T  & SV_{BB}S^\T
    \end{array}
\right].
\end{equation}

Now, let us choose the Gaussian unitaries 
$U_1$ with the symplectic matrix $S$ and $d_1=\sqrt{2}d$, and $U_2$ with the symplectic matrix $S$ and $d_2=0$. Then the action of the right-hand side on the first two moments will be
 \begin{equation}
\left[
\begin{array}{c}
    \bar{x}_A  \\
    \bar{x}_B
\end{array}
\right]
\to 
\left[
\begin{array}{c}
     S\left(\frac{\bar{x}_A +\bar{x}_B}{\sqrt{2}}\right)+d_1 \\
      S\left(\frac{\bar{x}_B -\bar{x}_A}{\sqrt{2}}\right)+d_2
\end{array}
\right]    
=\left[
\begin{array}{c}
     S\left(\frac{\bar{x}_A +\bar{x}_B}{\sqrt{2}}\right)+\sqrt{2}d \\
      S\left(\frac{\bar{x}_B -\bar{x}_A}{\sqrt{2}}\right)
\end{array}
\right],    
\end{equation}
and 
\begin{equation}
\left[
\begin{array}{cc}
  V_{AA}   & V_{AB}  \\
  V_{BA}   & V_{BB}
\end{array}
\right]
\to 
\left[
\begin{array}{cc}
  SV_{AA}S^\T   & SV_{AB}S^\T  \\
  SV_{BA} S^\T  & SV_{BB}S^\T
\end{array}
\right].
\end{equation}
Hence, for these chosen Gaussian unitaries $U_1$ and $U_2$, we obtain the result.

\end{proof}

\begin{prop}[Restatement of Proposition \ref{prop:non_entrop}]
The $(\alpha,k)$ R\'enyi non-Gaussian entropy $\NGE_{\alpha,k}(\psi)$  with $\alpha\in [0,+\infty]$ and $k\in \mathbb{Z}_+$  satisfies  the following properties:
 \begin{enumerate}[(1)]
    \item Faithfulness: $\NGE_{\alpha,k}(\psi)\geq 0$ with equality iff $\psi$ is a Gaussian state. 
    \item Gaussian invariance: $\NGE_{\alpha,k}(U_G\proj{\psi} U^\dag_G)=\NGE_{\alpha,k}(\psi)$ for any Gaussian unitary $U_G$. 
    \item Additivity under tensor product: $\NGE_{\alpha,k}(\psi_1\ot\psi_2)=\NGE_{\alpha,k}(\psi_1)+\NGE_{\alpha,k}(\psi_2)$. 
\end{enumerate}    
\end{prop}
\begin{proof}
(1) The faithfulness comes from the key observation that $\boxplus^{k} \psi$ is a pure state iff 
$\psi$ is a Gaussian state.

(2)
By Lemma \ref{lem:GU_Com}, for any Gaussian unitary $U_G$,  there exist
Gaussian unitaries  $U_1$ and $U_2$ such that 
\begin{align}
    U_{\pi/4}(U_G\otimes U_G)
    =(U_1\otimes U_2) U_{\pi/4},
\end{align}
where $U_{\pi/4}$ is the $50:50$ beam splitter.
This implies that 
\begin{align}
       U_{\pi/4}(U_G\otimes U_G)(\rho\otimes \rho)(U^\dag_G\otimes U^\dag_G)     U^\dag_{\pi/4}
       =(U_1\otimes U_2)U_{\pi/4}(\rho\otimes \rho)U^\dag_{\pi/4}(U^\dag_1\otimes U^\dag_2).
\end{align}
Hence, by the invariance of mutual information under product of unitaries, we get the result.

(3) Additivity under the tensor product comes from the fact that 
\begin{align}
    (\psi_1\ot\psi_2)\boxplus
     (\psi_1\ot\psi_2)
     =(\psi_1\boxplus\psi_1)\otimes (\psi_2\boxplus \psi_2).
\end{align}
    
\end{proof}

\section{The results for pure states with zero mean}\label{appen:zero_mean}
Now, let us consider a special case where the given pure bosonic state $\psi$ has zero mean, i.e., the first moment $\bar{x}=0$. In this case, the protocol will be simplified by the following result. 
\begin{prop}[Zero-mean case]\label{prop:zero-mean}
  Given a  pure bosonic state $\psi$ with mean vector  $\bar{x}=0$,
  $\psi$ is a Gaussian state iff $\psi\boxplus\psi=\psi$.
\end{prop}
\begin{proof}
On one hand, if the pure state  $\psi$ is a zero-mean Gaussian state, then by checking the characteristic function, it is direct to see that $\psi\boxplus\psi=\psi$.

On the other hand, given a pure state $\psi$ with zero mean, if 
it satisfies that  $\psi\boxplus\psi=\psi$, 
then by the proof in Theorem \ref{thm:key}, 
$\psi$ is a Gaussian state.

\end{proof}
Based on the ideas from the bosonic Gaussianity test for general pure states, we propose the following protocol for zero-mean states, utilizing three copies of the input states.

\begin{center}
  \begin{tcolorbox}[width=8cm,height=3.5cm,title=Bosonic Gaussian Test For Zero-mean State]
(1) Prepare 2 copies of the pure states $\psi$ and perform the 50:50 beam splitter to get a copy of  $\psi\boxplus\psi$.

(2) Perform the 50:50 beam splitter on $\psi\boxplus\psi$ and another copy of  $\psi$, and measure the parity on the second half of the system.
\end{tcolorbox}
\end{center}
(See also the experiment setup protocol displayed in Fig.~\ref{fig:BS_protocol_zero}).
The average parity measured in the many-body bosonic interference experiment is 
\begin{align}
\langle \tilde{P}\rangle
=\bra{\psi}\psi\boxplus\psi\ket{\psi}.
\end{align}
Here we use $\tilde{P}$ for parity to differentiate it from the parity $P$ used in the previous protocol.

\begin{Rem}
For a given zero-mean pure state $\psi$, the average purity obtained using four copies of the input state is given by
$$\langle P\rangle=\Tr{(\psi\boxplus\psi)^2}.$$
Conversely, when using three copies, the average purity is 
$$\langle \tilde{P}\rangle=\bra{\psi}\psi\boxplus\psi\ket{\psi}.$$
While both quantities are capable of detecting Gaussianity, they are distinct.

Now, let us discuss the relation between these two quantities for a zero-mean pure state. On one hand, by the Cauchy-Schwarz inequality, we have 
$\bra{\psi}\psi\boxplus\psi\ket{\psi}
\leq \sqrt{\Tr{\psi\boxplus\psi^2}}$, that is, 
$$\langle \tilde{P}\rangle\leq \sqrt{\langle P\rangle}.$$
On the other hand, we find that there is no general inequality of the type
$\langle P\rangle\leq c\langle \tilde{P}\rangle^b$ with $c>0$ and $b>0$, i.e., $\Tr{\psi\boxplus\psi^2}\leq c \bra{\psi}\psi\boxplus\psi\ket{\psi}^b$. For example, let us consider the number state $\ket{1}$, its self-convolution is
$\proj{1}\boxplus\proj{1}=\frac{1}{2}(\proj{0}+\proj{2})$. For this state,
$\Tr{\proj{1}\boxplus\proj{1}^2}=1/2$, whereas $\bra{1}(\proj{1}\boxplus\proj{1})\ket{1}=0$, thereby disproving the existence of such a lower bound.
\end{Rem}

\begin{figure}[h!]
    \centering
    \includegraphics[width=0.25 \linewidth]{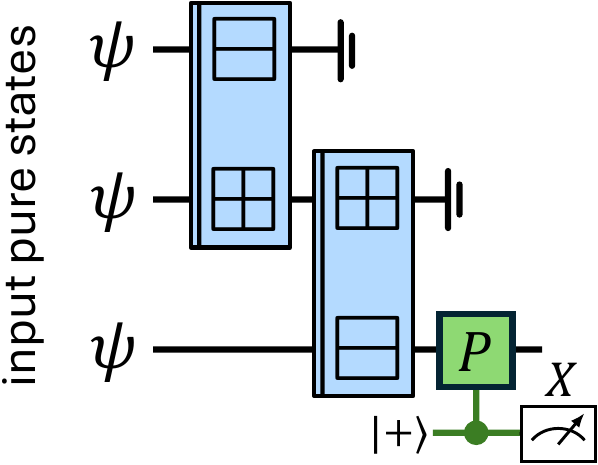}
    \caption{The protocol to test Gaussianity for pure input states with zero mean. }
    \label{fig:BS_protocol_zero}
\end{figure}

\section{Proof of the main results for mixed states}\label{appen:B}
\begin{prop}\label{thm:main_mix_S}
The Frobenius measure $d_F(\rho)$ satisfies the following properties: 
\begin{enumerate}[(1)]
    \item Faithfulness: $d_F(\rho)\geq 0$ with equality iff $\rho$ is a Gaussian state. 
    \item Gaussian invariance: $d_F(U_G\rho U^\dag_G) = d_F(\rho)$ 
    for any Gaussian unitary $U_G$.
    \item Subadditivity under tensor product: $d_F(\rho\otimes \sigma)\leq d_F(\rho)+ d_F(\sigma)$. 
\end{enumerate}
\end{prop}
\begin{proof}
    Faithfulness comes from Theorem \ref{thm:mix}. 
   Gaussian invariance stems from the commutation relation between a Gaussian unitary and a 50:50 beam splitter, as shown in Lemma~\ref{lem:GU_Com}, combined with the unitary invariance of the Frobenius norm.

    Subadditivity under the tensor product arises from the triangle inequality of the Frobenius norm. Specifically,
    if we define $\rho_{AB}=U_{\pi/4}(\rho\otimes \rho) U^\dag_{\pi/4}$, and 
    $\sigma_{AB}=U_{\pi/4}(\sigma\otimes \sigma) U^\dag_{\pi/4}$, then 
    \begin{align*}
        d_F(\rho\otimes \sigma)
        &=\norm{\rho_{AB}\otimes \sigma_{AB}-\rho_A\ot\rho_B\otimes \sigma_A\otimes \sigma_B}_F\\
        &\leq \norm{\rho_{AB}\otimes \sigma_{AB}-\rho_A\ot\rho_B\otimes \sigma_{AB}}_F
        +\norm{\rho_A\ot\rho_B\otimes \sigma_{AB}-\rho_A\ot\rho_B\otimes \sigma_A\otimes \sigma_B}_F\\
        &\leq  \norm{\rho_{AB}-\rho_A\ot\rho_B}_F
        +\norm{\sigma_{AB}-\sigma_A\otimes \sigma_B}_F\\
        &= d_F(\rho)+d_F(\sigma)\;.
    \end{align*}
\end{proof}

Let us consider a Gaussian channel
$\Lambda_G(\rho)=\Ptr{E}{U_G(\rho\otimes \sigma_E) U^\dag_G}$,
where the Gaussian state 
$\sigma_E$ is characterized by the 
first and second moments $(\bar{x}_B, V_{BB})$,
and the joint Gaussian unitary $U_G$ is characterized by 
\begin{equation}\label{eq:G_Ac}
   S=
   \left[
    \begin{array}{cc}
        S_{AA} & S_{AB} \\
        S_{BA} & S_{BB}
    \end{array}
    \right],\quad
d=
\left[
\begin{array}{cc}
      d_x \\
    d_y
    \end{array}
\right].
\end{equation}
Hence, the Gaussian channel $\Lambda_G$ is characterized by 
\begin{align}\label{eq:GC_Q}
    \bar{x}_A\to T\bar{x}_A+d,\quad
    V_{AA}\to TV_{AA}T^\T+N,
\end{align}
where $T=S_{AA}$, 
$d=S_{AB}\bar{x}_B+d_x$, and $N=S_{AB}V_{BB}S^\T_{AB}$.

\begin{lem}\label{lem:commu_GC}
Given a Gaussian channel $\Lambda_G$, there exist Gaussian channels $\Lambda_1$ and $\Lambda_2$  that satisfy the following commutation relation:
\begin{align}
    U_{\pi/4} \Lambda_G\ot\Lambda_G(\sigma_{AB}) U^\dag_{\pi/4}
    =\Lambda_1\ot\Lambda_2(U_{\pi/4}\sigma_{AB}U^\dag_{\pi/4}),
\end{align}
where  $\sigma_{AB}$ is any joint bosonic quantum state, and $U_{\pi/4}$ is the $50:50$ beam splitter.
\end{lem}
\begin{proof}
Since a Gaussian channel is fully characterized by how it transforms the first two moments, we only need to compare the effect of the left-hand side and right-hand side on these moments.

Let us consider the action of the given Gaussian channel $\Lambda_G(\rho)=\Ptr{E}{U_G\rho\otimes \sigma_E U^\dag_G}$.
Here, $\sigma_E$ is a Gaussian state characterized by its first and second moments $(\bar{x}_y, V_y)$,
and the action of the Gaussian unitary is described in Eq.~\eqref{eq:G_Ac}.
Hence, the overall action of this Gaussian channel on the first two moments is specified by $(T,d, N)$ in Eq.~\eqref{eq:GC_Q}. For the left-hand side of the relation we are considering, its action on the first two moments can be described as follows:
\begin{equation}
\left[
\begin{array}{c}
      \bar{x}_A \\
     \bar{x}_B
\end{array}
\right]
\to 
\left[
\begin{array}{c}
     T( \frac{\bar{x}_A+\bar{x}_B}{\sqrt{2}}) +\sqrt{2}d\\
   T( \frac{\bar{x}_B-\bar{x}_A}{\sqrt{2}} )
\end{array}
\right],
\end{equation}
and 
\begin{equation}
\left[
\begin{array}{cc}
   V_{AA}  & V_{AB}  \\
    V_{BA} & V_{BB} 
\end{array}
\right]
\to 
\frac{1}{2}
\left[
\begin{array}{cc}
  T (V_{AA} +V_{AB}+V_{BA}+V_{BB}) T^\T  & T(-V_{AA} +V_{AB}-V_{BA}+V_{BB}) T^\T  \\
   T (-V_{AA} -V_{AB}+V_{BA}+V_{BB})T^\T  & T(V_{AA} -V_{AB}-V_{BA}+V_{BB}) T^\T 
\end{array}
\right]
+
\left[
\begin{array}{cc}
N&0 \\
0  & N
\end{array}
\right].
\end{equation}

Now, let us assume that the action of  Gaussian channels $\Lambda_1$ and $\Lambda_2$ on the first two moments is described as $(T,d_1, N)$  and $(T,d_2, N)$ in Eq.~\eqref{eq:GC_Q}, where $d_1=\sqrt{2}d$ and $d_2=0$. Hence, for the right-hand side, its action on the first two moments can be described as follows:
\begin{equation}
\left[
\begin{array}{c}
      \bar{x}_A \\
     \bar{x}_B
\end{array}
\right]
\to 
\left[
\begin{array}{c}
     T( \frac{\bar{x}_A+\bar{x}_B}{\sqrt{2}}) +d_1\\
   T( \frac{\bar{x}_B-\bar{x}_A}{\sqrt{2}} )+d_2
\end{array}
\right]
=
\left[
\begin{array}{c}
     T( \frac{\bar{x}_A+\bar{x}_B}{\sqrt{2}}) +\sqrt{2}d\\
   T( \frac{\bar{x}_B-\bar{x}_A}{\sqrt{2}} )
\end{array}
\right],
\end{equation}
and 
\begin{equation}
\left[
\begin{array}{cc}
   V_{AA}  & V_{AB}  \\
    V_{BA} & V_{BB} 
\end{array}
\right]
\to 
\frac{1}{2}
\left[
\begin{array}{cc}
  T (V_{AA} +V_{AB}+V_{BA}+V_{BB}) T^\T  & T(-V_{AA} +V_{AB}-V_{BA}+V_{BB}) T^\T  \\
   T (-V_{AA} -V_{AB}+V_{BA}+V_{BB})T^\T  & T(V_{AA} -V_{AB}-V_{BA}+V_{BB}) T^\T 
\end{array}
\right]
+
\left[
\begin{array}{cc}
N&0 \\
0  & N
\end{array}
\right].
\end{equation}
Hence, the left-hand side is equal to the right-hand side
by choosing Gaussian channels  $\Lambda_1$ and $\Lambda_2$ with the parameters $  (T,\sqrt{2}d, N)$  and $(T,0, N)$.

Finally, we can construct the Gaussian channels $\Lambda_1$ and $\Lambda_2$. Let us consider the Gaussian state $\sigma_1$ with the first two moments $(\sqrt{2}\bar{y}, V_y)$ and the Gaussian unitary $U_1$ with the symplectic matrix $S$ and $d'=\sqrt{2}d$. Hence,  the Gaussian channel $\Lambda_1(\rho)=\Ptr{E}{U_1(\rho\otimes \sigma_1) U^\dag_1}$ has the parameters
$  (T,\sqrt{2}d, N)$. 
Similarly, we consider 
$\sigma_2$ with the first two moments $(0, V_y)$ and the 
Gaussian unitary $U_2$ with the symplectic matrix $S$ and $d'=0$.
Hence,  the Gaussian channel $\Lambda_2(\rho)=\Ptr{E}{U_2(\rho\otimes \sigma_2) U^\dag_2}$  has the parameters
$(T,0,N)$. 
\end{proof}

\begin{Def}
Given an $N$-mode quantum state $\rho$, the $\alpha$-mutual information of non-Gaussianity 
$\MING_{\alpha}$ is 
\begin{equation}\label{eq:MING_def}
    \MING_{\alpha}(\rho)
    =D_{\alpha}(\rho_{AB}||\rho_A\ot\rho_B)\;,
\end{equation}
where $\rho_{AB}=U_{\pi/4}\rho\otimes \rho U^\dag_{\pi/4}$, and  $ D_{\alpha}(\rho||\sigma):=\frac{1}{\alpha-1}\log_2\Tr{\left(\sigma^{\frac{1-\alpha}{2\alpha}}\rho\sigma^{\frac{1-\alpha}{2\alpha}}\right)^{\alpha}}$ is the sandwich quantum R\'enyi relative entropy. When $\alpha=1$, we refer to it as the mutual information of non-Gaussianity (or Ming for short), which can also be written as 
\begin{align}\label{eq:MING}
  \MING(\rho)=D(\rho_{AB}||\rho_A\ot\rho_B)
  =S(\rho\boxplus\rho)+S(\rho\boxminus\rho)-2S(\rho).
\end{align}
\end{Def}

\begin{prop}
The $\alpha$-Mutual-Information of Non-Gaussianity $\MING_{\alpha}$ for $\alpha\geq 1/2$ satisfies the following properties: 
\begin{enumerate}[(1)]
    \item Faithfulness: $\MING_{\alpha}(\rho)\geq 0$ with equality iff $\rho$ is a Gaussian state. 
    \item Monotonicity under Gaussian channel: $\MING_{\alpha}(\Lambda_G(\rho))\leq \MING_{\alpha}(\rho)$ for any Gaussian channel $\Lambda_G$.
    \item Additivity under tensor product: $\MING_{\alpha}(\rho\otimes \sigma)=\MING_{\alpha}(\rho)+\MING_{\alpha}(\sigma)$. 
\end{enumerate}
\end{prop}
\begin{proof}
Properties (1) and (3) come from the definition. Now, let us focus on the property (2). Recall that, given two copies of the input state $\rho$, the output state after the $50:50$ beam splitter is 
\begin{align}
    \rho_{AB}=
    U_{\pi/4} (\rho\otimes \rho) U^\dag_{\pi/4},
\end{align}
For the two copies of the state $\Lambda_G(\rho)$, the output state after the $50:50$ beam splitter is 
\begin{align}
    \sigma_{AB}
    =U_{\pi/4} \Lambda_G(\rho)\ot\Lambda_G(\rho) U^\dag_{\pi/4}
    =\Lambda_1\ot\Lambda_2(U_{\pi/4}(\rho\otimes \rho) U^\dag_{\pi/4})
    =\Lambda_1\ot\Lambda_2(\rho_{AB}),
\end{align}
where the second equality comes from Lemma \ref{lem:commu_GC}. Hence, by the monotonicity of quantum relative entropy under a quantum channel, we have $\MING_{\alpha}(\Lambda_G(\rho))\leq \MING_{\alpha}(\rho)$.

\end{proof}

It is worth noting that for a pure state $\psi$, the quantity $\MING(\psi)$ defined in Eq.~\eqref{eq:MING} reduces to the non-Gaussian entropy $\NGE_{1,1}(\psi)$.
This follows from the fact that the von Neumann entropy of a pure state is zero, and the states $\rho \boxplus \rho$ and $\rho \boxminus \rho$ share the same spectrum. This observation confirms that $\MING(\rho)$ serves as a natural generalization of non-Gaussian entropy to mixed states, thereby extending the framework for quantifying non-Gaussianity.
Fig.~\ref{fig:curves}(c) and (d) illustrate the numerical value of $\MING(\rho)$ for mixed states of the form $\rho = \CNN_L[\gamma_L](\psi)$, where $\psi$ is chosen to be either a cat state or a Fock state. As expected, the loss channel $\CNN_L$ (which is a Gaussian channel) reduces non-Gaussianity, driving the state toward the vacuum as $\gamma_L \to 1$. Analytical expressions are provided in Appendix~\ref{appen:exam}. 

\begin{figure}[h!]
    \centering
    \includegraphics[width=0.5 \linewidth]{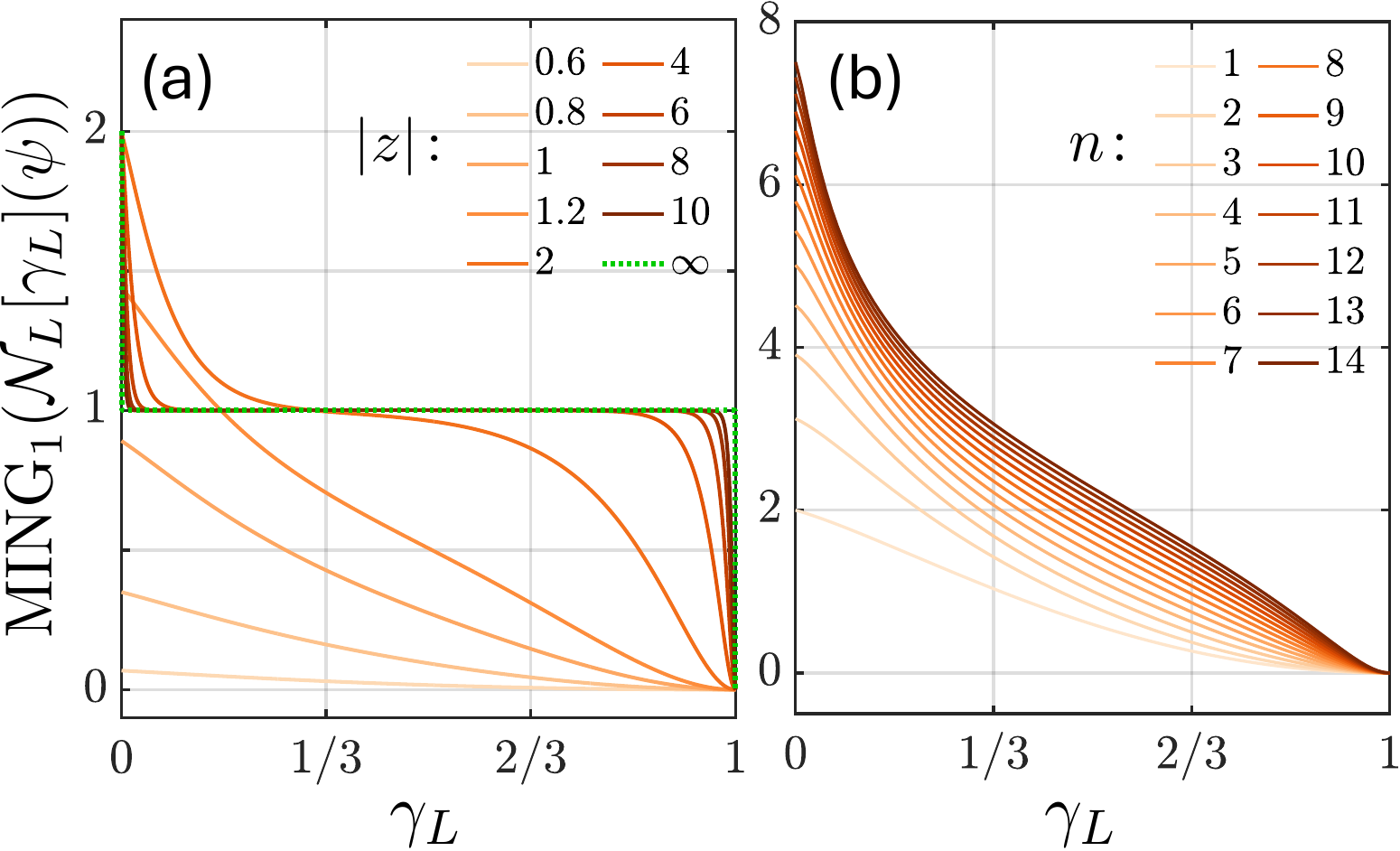}
    \label{fig:MING_curves}
    \caption{This figure demonstrates the values of $\MING_1(\rho)$ (see Eq.~\eqref{eq:MING_def}, where $\rho$ is a mixed state $\CNN_L[\gamma_L](\psi)$. (a) $\psi\propto \ket{z}+\ket{-z}$ is the two component cat state. (b) $\psi = \proj{n}$ is the Fock state of particle number $n$.}
\end{figure}

\section{Noisy bosonic channels}
In this work, we simulate the protocol in a noisy setting. For instance, the beam splitter is assumed to be noisy by introducing uncertainty on the rotational angle:
\begin{equation}\label{eq:Btheta}
    \CBB_\theta[\sigma_B^2](\rho) := \frac{1}{\sqrt{2\pi\sigma_B^2}}\int_{-\infty}^{\infty} e^{-\frac{\varphi^2}{2\sigma_B^2}} U_{\theta+\varphi}\rho U_{\theta + \varphi}^\dag \mathrm{d} \phi\;.
\end{equation}
For simplicity, we consider each bosonic mode suffers from a random displacement channel $\CNN_{D}[\sigma_D^2]$, a dephasing channel $\CNN_{P}[\sigma_P^2]$, and a bosonic loss channel $\CNN_{L}[\gamma_L]$, which are defined as:
\begin{align}
    \CNN_{D}[\sigma_D^2](\rho) &:= \frac{1}{2\pi \sigma_D^2}\int_{\complex} e^{-\frac{|\xi|^2}{2\sigma_D^2}} D(\xi) \rho D(\xi)^\dag \mathrm{d}^2 \xi \;,\label{eq:ND}\\
    \CNN_{P}[\sigma_P^2](\rho) &:= \frac{1}{\sqrt{2\pi\sigma_P^2}}\int_{-\infty}^{\infty} e^{-\frac{\varphi^2}{2\sigma_P^2}} e^{-i\varphi\hat{n}} \rho e^{i\varphi\hat{n}} \mathrm{d} \varphi \;,\label{eq:NP}\\
    \CNN_{L}[\gamma_L](\rho) &:=  \Ptr{E}{U_{\theta_L}(\rho\otimes \proj{0}_E)U_{\theta_L}^\dag}\;,\label{eq:NL}
\end{align}
where $\hat{n}:=\hat{a}^\dag \hat{a}$, $\sigma_D > 0$, $\sigma_P > 0$, $\gamma_L\in[0,1]$, and $\theta_L:=\arcsin\sqrt{\gamma_L}$. In the above definitions, $\CNN_D$ and $\CNN_P$ come from the uncertain physical unitary perturbations $D(\xi)$ and $e^{i\phi\hat{n}}$.  The definition of $\CNN_L[\gamma_L]$ and $\CNN_P[\sigma_P]$ can be considered as the integral consequence of the master equation $\frac{\mathrm{d}\rho}{\mathrm{d}t} = \kappa_L(\hat{a}\rho \hat{a}^\dag - \frac{1}{2}\{\hat{a}^\dag \hat{a}, \rho\})$, with $\gamma_L \equiv 1-e^{-\kappa_L t}$.
Note that the following Gaussian integral regarding a Hermitian operator $A$:
\begin{equation}
    \rho(s)=\frac{1}{\sqrt{2\pi s}} \int_{-\infty}^{\infty} e^{-\frac{\varphi^2}{2s}} e^{-i\varphi A}\rho(0) e^{i\varphi A^\dag} \mathrm{d} \varphi
\end{equation}
is the solution of the following differential equation at $t>0$:
\begin{equation}
    \frac{\mathrm{d}\rho}{\mathrm{d}s} = A \rho A^\dag - \frac{1}{2}\{A^\dag A, \rho\}\;.
\end{equation}
Thus, the action of $\CNN_P[\sigma_P]$ can be considered as the integral of a master equation $\frac{\mathrm{d}\rho}{\mathrm{d}t} = \kappa_P(\hat{n}\rho \hat{n}^\dag - \frac{1}{2}\{\hat{n}^2, \rho\})$, with $\sigma_P^2 \equiv \kappa_P t$. Similarly, the action of $\CNN_D[\sigma_D]$ can be considered as the integral of a master equation $\frac{\mathrm{d}\rho}{\mathrm{d}t} = \kappa_P(\hat{J}\rho \hat{J}^\dag - \frac{1}{2}\{\hat{J}^2, \rho\})$, with $\sigma_D^2 \equiv \kappa_D t$ and Hermitian jump operator $\hat{J} = \hat{x}^\T\Omega \xi = (\hat{q},\hat{p})\Omega (\xi_1, \xi_2)^\T$. 
It can be shown that
\begin{align}
    \CNN_{D}[\sigma_D^2] \circ \CNN_{P}[\sigma_P^2] &= \CNN_{P}[\sigma_P^2]\circ \CNN_{D}[\sigma_D^2] \;, \\
    \CNN_{L}[\gamma_L] \circ \CNN_{P}[\sigma_P^2] &= \CNN_{P}[\sigma_P^2]\circ \CNN_{L}[\gamma_L] \;,\\
    \CNN_{L}[\gamma_L] \circ \CNN_{D}[\sigma_D^2] &= \CNN_{D}[(1-\gamma_L)\sigma_D^2]\circ \CNN_{L}[\gamma_L]\;, 
\end{align}
and
\begin{align}
    \CNN_{D}[\sigma_1^2] \circ \CNN_{D}[\sigma_2^2] &= \CNN_{D}[\sigma_1^2 + \sigma_2^2]\;, \\
    \CNN_{P}[\sigma_1^2] \circ \CNN_{P}[\sigma_2^2] &= \CNN_{P}[\sigma_1^2 + \sigma_2^2]\;, \\
    \CNN_{L}[\gamma_1] \circ \CNN_{L}[\gamma_2] &= \CNN_{L}[\gamma_1 + \gamma_2 - \gamma_1\gamma_2]\;. 
\end{align}
In the simulation of the $\NGE_{2,1}$ measurement protocol, we define the noise channel for the idle bosonic mode as
\begin{equation}
    \CNN[\sigma_D^2,\sigma_P^2,\gamma_L]:=\CNN_D[\sigma_D^2]\circ \CNN_P[\sigma_P^2]\circ \CNN_L[\gamma_L]\;.
\end{equation}
The above ordering is just for convenience, as one can adjust the parameters if a different ordering is adopted. For example, utilizing the above commuting rules, we obtain
\begin{equation}
    \lim_{n\to \infty}\prod_{j=1}^{n} \CNN_D\left[\frac{\kappa_D\tau}{n}\right] \circ \CNN_P\left[\frac{\kappa_P \tau}{n}\right] \circ \CNN_L\left[\frac{\kappa_L\tau}{n}\right] = \CNN_D\left[\left(1-e^{-\kappa_L\tau}\right)\frac{\kappa_D}{\kappa_L}\right] \circ \CNN_P\left[\kappa_P \tau\right] \circ \CNN_L\left[1-e^{-\kappa_L\tau}\right]\;,
\end{equation}
for some $\tau > 0$. The above infinite product with infinitesimal parameter implies that the action of $\CNN[\sigma_D^2,\sigma_P^2,\gamma_L]$ can be integrated from the following Lindblad master equation:
\begin{equation}
    \frac{\mathrm{d}\rho}{\mathrm{d}\tau} = \kappa_L\left(\hat{a}\rho\hat{a}^\dagger - \frac{1}{2}\{\hat{n},\rho\}\right) - \frac{\kappa_D}{2}\left[\hat{n},\left[\hat{n},\rho\right]\right] - \frac{\kappa_P}{2}\left(\left[\hat{q},\left[\hat{q},\rho\right]\right]+\left[\hat{p},\left[\hat{p},\rho\right]\right]\right)
\end{equation}
with jump operators $\hat{n}$, $\hat{p}$, $\hat{q}$, $\hat{a}$, and parameters $\sigma_D^2 = \kappa_D\tau$, $\sigma_P^2 = \kappa_P\tau$, and $\gamma_L = 1 - e^{-\kappa_L\tau}$.

\section{Analytical calculation of non-Gaussianity measures}\label{appen:exam}
\subsection{Example: Fock state}
We denote $\ket{n}$ ($n \in \mathbb{N}$) as the Fock state of a single bosonic mode, which satisfies $\hat{n}\ket{n} = n\ket{n}$. It is known that, except for the vacuum state $\ket{0}$, all other states are non-Gaussian. This section calculates the non-Gaussianity measures $\NGE_{\alpha,1}(\proj{n})$, $d_F(\CNN_L[\gamma_L](\psi))$ and $\MING(\CNN_L[\gamma_L](\psi))$ for $\psi = \proj{n}$.

First, utilizing the celebrated Jordan–Schwinger map (sometimes known as the Schwinger boson representation), we can use the $SU(2)$ algebra to represent some operators of two bosonic modes:
\begin{align}
    S_+ \simeq \hat{a}_A^\dag \hat{a}_B, \quad
    S_- \simeq \hat{a}_A \hat{a}_B^\dag,\quad
    S_z \simeq \frac{1}{2}\left(\hat{a}_A^\dag \hat{a}_A - \hat{a}_B^\dag \hat{a}_B\right)\;.
\end{align}
If we denote $S_x := \frac{S_+ + S_-}{2}$ and $S_y := \frac{S_+ - S_-}{2i}$, then there will be $[S_\alpha, S_\beta] = i \epsilon_{\alpha\beta\gamma}S_\gamma$ for all $\alpha,\beta,\gamma\in\{x,y,z\}$ and $\epsilon_{\alpha\beta\gamma}$ is the Levi-Civita symbol.
Thus, the action of the beam splitter $U_\theta = e^{\theta(\hat{a}_A^\dag \hat{a}_B - \hat{a}_A \hat{a}_B^\dag)}$ can be represented as $e^{i2\theta S_y}$.
Note that the beam splitter unitary $U_\theta$ commutes with $\hat{n}_A + \hat{n}_B$, so if the initial state $\ket{\psi}$ is an eigenstate of $\hat{n}_1 + \hat{n}_2$ with eigenvalue $2S$, $U_\theta\ket{\psi}$ can be treated as a spin state of spin $S$, rotated along $y$ axis by $2\theta$.
Concretely, the correspondence between $\ket{S=m_z}_S$ under spin-$S$ and the 2-mode Fock state $\ket{m_A}\otimes\ket{m_B}$ is given by:
\begin{equation}
    m_A = S + m_z,\quad m_B = S - m_z\;,
\end{equation}
which implies $S = \frac{m_A + m_B}{2}$.
To find the matrix representation of the $50:50$ beam splitter, we need to find the unitary at $\theta = \pi/2$. We find that 
\begin{equation}
    \bra{S_z = m_z}e^{i\frac{\pi}{2}S_y}\ket{S_z = 0}_S = \left\{
    \begin{array}{ll}
         (-1)^{\frac{S-m_z}{2}}\frac{\sqrt{(S+m_z)!(S-m_z)!}}{2^S\left(\frac{S+m_z}{2}\right)!\left(\frac{S-m_z}{2}\right)!}, & \text{if }S+m_z\text{ is even},  \\
         0, & \text{otherwise}\;.
    \end{array}\right.
\end{equation}
When the input state of the beam splitter is $\ket{n}\otimes \ket{n}$, we have $S = n$ and the input state is mapped to $\ket{S_z = 0}$. Thus, $\proj{n}\boxplus\proj{n}$ is diagonal in the Fock basis with the diagonal entry $\bra{S_z = m_z}e^{i\frac{\pi}{2}S_y}\ket{S_z = 0}^2$:
\begin{equation}
    \bra{m}\left(\proj{n}\boxplus\proj{n}\right)\ket{m}
    =\left\{
    \begin{array}{ll}
         \frac{m!(2n-m)!}{\left(2^{n}\left(\frac{m}{2}\right)!\left(n-\frac{m}{2}\right)!\right)^2}, &  \text{if }m\text{ is even and } m\le 2n\;, \\
         0, &\text{otherwise}\;, 
    \end{array}
    \right.
\end{equation}
where we have replaced $S+m_z = m$ and $S-m_z = 2n - m$. The value of $m$ can be any natural number.
Thus, the non-Gaussian entropy is given by
\begin{equation}
    \NGE_{\alpha,1}(\proj{n}) = \frac{1}{1-\alpha}\log_2\sum_{k=0}^n \frac{\left((2k)!(2n-2k)!\right)^\alpha}{\left(2^{n} k!\left(n-k\right)!\right)^{2\alpha}}\;,
\end{equation}
and the mutual information non-Gaussianity is
\begin{equation}
   \MING(\proj{n}) = -\sum_{k=0}^n \frac{(2k)!(2n-2k)!}{\left(2^{n} k!\left(n-k\right)!\right)^2} \log_2 \frac{(2k)!(2n-2k)!}{\left(2^{n} k!\left(n-k\right)!\right)^2}\;.
\end{equation}

Similarly, let $\CNN_L[\gamma_l]$ be the pure loss channel with $\gamma_L\in[0,1]$ (see Eq.~\eqref{eq:NL}). For convenience, let us define
\begin{equation}
     U^{(S)}_{m_{z},m_{z}'}(\theta):=
        \bra{S_z=m_{z}}e^{i2\theta S_y}\ket{S_z = m_{z}'}_S\;,
\end{equation}
which can be explicitly given by the Wigner D-matrix.
Since the initial is the Fock state $\psi = \proj{n}\otimes \proj{0}$, we can replace $S = n/2$, $m_z = S - m$ and obtain the noisy input state
\begin{equation}\label{eq:NL_nn}
    \rho:=\CNN_L[\gamma_L](\proj{n}) = \sum_{m=0}^n \left|U^{(n/2)}_{\frac{n}{2}-m,\frac{n}{2}}(\theta_L)\right|^2  \proj{m}\;.
\end{equation}
The next step is finding $\proj{m_A}\boxplus \proj{m_B}$ for all $0\le m_A,m_B\le n$, we have
\begin{equation}
    \proj{m_A}\boxplus\proj{m_B} = \sum_{\ell = 0}^{m_A+m_B} \left|U^{\left((m_A+m_B/2\right)}_{\frac{m_A+m_B}{2} - \ell,\frac{m_A-m_B}{2}}\left(\frac{\pi}{4}\right)\right|^2 \proj{\ell}\;.
\end{equation}
Thus, we arrive at
\begin{equation}\label{eq:NL_nn_boxplus}
        \rho \boxplus \rho = \sum_{m_A,m_B = 0}^n  \sum_{\ell = 0}^{m_A+m_B} \left|U^{(n/2)}_{\frac{n}{2}-m_A,\frac{n}{2}}(\theta_L) U^{(n/2)}_{\frac{n}{2}-m_B,\frac{n}{2}}(\theta_L) U^{\left((m_A+m_B/2\right)}_{\frac{m_A+m_B}{2} - \ell,\frac{m_A-m_B}{2}}\left(\frac{\pi}{4 }\right)\right|^2 \proj{\ell}\;.
\end{equation}
Since $\rho\boxminus \rho = \rho\boxplus\rho$, the mutual information of non-Gaussianity is given by
\begin{equation}
    \MING(\rho) = 2S(\rho\boxplus\rho) - 2S(\rho)\;,
\end{equation}
where the involved states are diagonal in the Fock basis, as shown in Eq.~\eqref{eq:NL_nn} and \eqref{eq:NL_nn_boxplus}. The numerical result is illustrated in Fig.~\ref{fig:MING_curves}(d). 

To calculate $d_F(\rho)$, we can expand it as:
\begin{equation}
    \begin{aligned}
        d_F(\rho) &= \|U_{\pi/4}(\rho\otimes\rho)U_{\pi/4}^\dag - (\rho\boxplus\rho)(\rho\boxminus\rho)\|_F \\
        &= \left\{\Tr{\rho^2}^2 + \Tr{(\rho\boxplus\rho)^2}\Tr{(\rho\boxminus\rho)^2} - 2\Tr{U_{\pi/4}(\rho\otimes\rho)U_{\pi/4}^\dag (\rho\boxplus\rho)\otimes (\rho\boxminus\rho)}\right\}^{1/2} \\
        &= \left\{\Tr{{\rho^2}}^2 + \Tr{(\rho\boxplus\rho)^2}^2 - 2\Tr{U_{\pi/4}(\rho\otimes\rho)U_{\pi/4}^\dag (\rho\boxplus\rho)\otimes (\rho\boxplus\rho)}\right\}^{1/2} \;.
    \end{aligned}
\end{equation}
Note that $\rho\otimes\rho$ usually has a high rank, the explicit expression of $- 2\Tr{U_{\pi/4}(\rho\otimes\rho)U_{\pi/4}^\dag (\rho\boxplus\rho)\otimes (\rho\boxplus\rho)}$ will be cumbersome. However, since $U_{\pi/4}(\rho\otimes\rho)U_{\pi/4}^\dag$ can be exactly represented in the Fock state with finite cutoff, we can still calculate $d_F(\rho)$ under any $\gamma_L$ without numerical approximation.

\subsection{Example: two-component cat state}
We demonstrate the derivation for the non-Gaussianity measure for the noisy two-component cat state: $\CNN_L[\gamma_L](\psi)$, where $\ket{\psi}\propto \ket{z}+\ket{-z}$ ($z\in\mathbb{C}$), where $\ket{z} := e^{z\hat{a}^\dag - z^* \hat{a}}\ket{0}$ is the coherent state.
To start with, we define the \textit{un-normalized} two-component cat state as:
\begin{equation}\label{eq:catplus}
    \phi_{\pm}(z) := \frac{1}{2}(\ket{z}\pm\ket{-z})(\bra{z}\pm\bra{-z})\;.
\end{equation}
The temporarily ignored normalization factor is denoted as:
\begin{equation}
    N_{\pm}:=\Tr{\phi_\pm(z)} = 1+e^{-2|z|^2}\;.
\end{equation}
In other words, the normalized two-component cat state is written as $\psi = \phi_+(z)/N_+$. 
To obtain the action of the bosonic loss channel $\CNN_L[\gamma_L]$ (Eq.~\eqref{eq:NL}), we first find the action of $U_{\theta_L}$
\begin{equation}
    U_{\theta_L}\frac{\ket{z} + \ket{-z}}{\sqrt{2}}\otimes \ket{0}
    =
    \frac{1}{\sqrt{2}}\left(\ket{\sqrt{1-\gamma_L}z}\ket{\sqrt{\gamma_L}z} + \ket{-\sqrt{1-\gamma_L}z}\ket{-\sqrt{\gamma_L}z}\right)\;.
\end{equation}
Let
\begin{equation}
    y := \sqrt{1-\gamma_L}z\;,\qquad n_{\pm} := \frac{1}{2}(1\pm e^{-2\gamma_L|z|^2})\;,
\end{equation}
then it is straightforward to see that the noisy input state is a low-rank state:
\begin{equation}\label{eq:Nphi_phi_phi}
    \begin{aligned}
        \rho&:=\frac{1}{N_+}\CNN_L[\gamma_L](\phi_+(z))\\
        &= \frac{1}{2}
        \begin{bmatrix}
            \ket{y} & \ket{-y}
        \end{bmatrix}
        \begin{bmatrix}
            1 & \langle \sqrt{\gamma_L} z| -\sqrt{\gamma_L} z\rangle\\
            \langle -\sqrt{\gamma_L} z| \sqrt{\gamma_L} z\rangle & 1\\
        \end{bmatrix}
        \begin{bmatrix}
            \bra{\sqrt{y}} \\ \bra{-y}
        \end{bmatrix}\\
        &=\frac{1}{2}
        \begin{bmatrix}
            \ket{y} & \ket{-y}
        \end{bmatrix}
        \begin{bmatrix}
            1 & e^{-2\gamma_L|z|^2}\\
            e^{-2\gamma_L|z|^2} & 1\\
        \end{bmatrix}
        \begin{bmatrix}
            \bra{y} \\ \bra{-y}
        \end{bmatrix} \\
        &=n_+\phi_+(y) + n_-\phi_-(y)\;.
    \end{aligned}
\end{equation}
The non-zero eigenvalues of $\rho$ are the same as the following matrix:
\begin{equation}
\begin{aligned}
    T_1 &:= \frac{1}{2N_+(z)} \begin{bmatrix}
            1 & e^{-2\gamma_L|z|^2}\\
            e^{-2\gamma_L|z|^2} & 1\\
        \end{bmatrix}^{1/2}
        \begin{bmatrix}
            \langle y|y\rangle & \langle y|-y\rangle\\
            \langle -y|y\rangle & \langle -y|-y\rangle\\
        \end{bmatrix}
        \begin{bmatrix}
            1 & e^{-2\gamma_L|z|^2}\\
            e^{-2\gamma_L|z|^2} & 1\\
        \end{bmatrix}^{1/2} \\
        &= \frac{1}{2(1+e^{-2|z|^2})} \begin{bmatrix}
            1 & e^{-2\gamma_L|z|^2}\\
            e^{-2\gamma_L|z|^2} & 1\\
        \end{bmatrix}^{1/2}
        \begin{bmatrix}
            1 & e^{-2|y|^2}\\
            e^{-2|y|^2} & 1\\
        \end{bmatrix}
        \begin{bmatrix}
            1 & e^{-2\gamma_L|z|^2}\\
            e^{-2\gamma_L|z|^2} & 1\\
        \end{bmatrix}^{1/2}\;.
\end{aligned}
\end{equation}
The next step is finding the convolution of $\rho$. 
Utilizing $\phi_+(y)\boxplus\phi_-(y) = \phi_-(y)\boxplus\phi_+(y)$ (due to the symmetry in the phase space), the convolution of the noisy input can be represented as:
\begin{equation}\label{eq:NLphi_plus_boxplus}
    \begin{aligned}
        \rho\boxplus \rho &= \frac{1}{N_+^2}\CNN_L[\gamma_L](\phi_+(z))\boxplus \CNN_L[\gamma_L](\phi_+(z))\\
        &= \frac{1}{N_+^2}\left(n_+^2 \phi_+(y)\boxplus\phi_+(y) + 2n_+n_- \phi_+(y)\boxplus\phi_-(y) + n_-^2 \phi_-(y)\boxplus\phi_-(y)\right)\;.
    \end{aligned}  
\end{equation}
To proceed, we need to find the convolution between $\phi_\pm(y)$. 
After some derivation, we obtain:
\begin{equation}
\begin{aligned}
    \phi_\pm(y)\boxplus \phi_\pm(y) & = \Ptr{2}{U_{\pi/4}(\phi_\pm(y) \otimes \phi_\pm(y))U_{\pi/4}^\dag}\\
    &= \Ptr{2}{\frac{\ket{0} \otimes \left(\ket{\sqrt{2}y} + \ket{-\sqrt{2}y}\right) \pm \left(\ket{\sqrt{2}y} + \ket{-\sqrt{2}y}\right)\otimes \ket{0}}{2}\times \text{h.c.} } \\
    &=\frac{1}{4}\begin{bmatrix}
        \ket{0} & \ket{\sqrt{2}y}+ \ket{-\sqrt{2}y}
    \end{bmatrix}
    \begin{bmatrix}
        2(1+e^{-4|y|^2}) & \pm 2 e^{-|y|^2}\\
        \pm 2 e^{-|y|^2} & 1
    \end{bmatrix}
    \begin{bmatrix}
        \bra{0} \\ \bra{\sqrt{2}y} + \bra{-\sqrt{2}y}
    \end{bmatrix} \;,
\end{aligned}    
\end{equation}
\begin{equation}
    \begin{aligned}
    \phi_\pm(y)\boxplus \phi_\mp(y) & = \Ptr{2}{U_{\pi/4}(\phi_\pm(y) \otimes \phi_\mp(y))U_{\pi/4}^\dag}\\
    &= \Ptr{2}{\frac{\ket{0} \otimes \left(\ket{\sqrt{2}y} - \ket{-\sqrt{2}y}\right) \pm \left(\ket{\sqrt{2}y} - \ket{-\sqrt{2}y}\right)\otimes \ket{0}}{2}\times \text{h.c.} } \\
    &=\frac{1}{4}\begin{bmatrix}
        \ket{0} & \ket{\sqrt{2}y} - \ket{-\sqrt{2}y}
    \end{bmatrix}
    \begin{bmatrix}
        2(1-e^{-4|y|^2}) & 0\\
        0 & 1
    \end{bmatrix}
    \begin{bmatrix}
        \bra{0} \\ \bra{\sqrt{2}y} - \bra{-\sqrt{2}y}
    \end{bmatrix}\;,
\end{aligned}    
\end{equation}
and 
\begin{equation}
    \phi_\pm(y)\boxminus \phi_\pm(y) = \phi_\pm(y)\boxplus \phi_\pm(y)\;,\quad
    \phi_\pm(y)\boxminus \phi_\mp(y) = \phi_\pm(y)\boxplus \phi_\mp(y)\;. 
\end{equation}
Denote that 
\begin{align}
    \Upsilon &:= 
    \begin{bmatrix}
        \ket{0} & \ket{\sqrt{2}y} + \ket{-\sqrt{2}y} & \ket{\sqrt{2}y} - \ket{-\sqrt{2}y}
    \end{bmatrix} \;, \\
    A &:= \frac{1}{4}\begin{bmatrix}
        2(n_+^2 + n_-^2)(1+e^{-4|y|^2}) + 4n_+n_-(1-e^{-4|y|^2})&2(n_+^2 - n_-^2)e^{-|y|^2}&0\\
        2(n_+^2 - n_-^2)e^{-|y|^2}&n_+^2 + n_-^2&0\\
        0&0&2n_+n_-
    \end{bmatrix}\;,
\end{align}
we have:
\begin{equation}
    \rho\boxplus \rho = \frac{1}{N_+^{-2}}\CNN_L[\gamma_L](\phi_+(z))\boxplus \CNN_L[\gamma_L](\phi_+(z)) = \frac{1}{N_+^{2}} \Upsilon A \Upsilon^\dag\;.
\end{equation}
It follows that the spectrum of the normalized state $\rho'$ is given by the eigenvalues of the following:
\begin{equation}
    \begin{aligned}
        T_2&:= \frac{1}{N_+^2}A^{1/2}\Upsilon^\dag \Upsilon A^{1/2} \\
        &= 
        \frac{1}{(1+e^{-2|z|^2})^2}
        A^{1/2}
        \begin{bmatrix}
            1&2e^{-|y|^2}&0\\
            2e^{-|y|^2}&2(1+e^{-4|y|^2})&0\\
            0&0&2(1-e^{-4|y|^2})
        \end{bmatrix}
        A^{1/2}\;.
    \end{aligned}
\end{equation}
Hence, due to the symmetry of $\boxplus$ and $\boxminus$, the mutual information non-Gaussianity of $\rho = \CNN_L[\gamma_L](\psi)$ is given by:
\begin{equation}\label{eq:MING_lossy_cat}
    \begin{aligned}
        \MING(\rho) &= -2\Tr{(\rho\boxplus\rho)\log_2 (\rho\boxplus\rho)} - 2\Tr{\rho\log_2 \rho} \\
        &= -2\Tr{T_2\log_2 T_2} -2\Tr{T_1\log_2 T_1}\;.
    \end{aligned}
\end{equation}
The numerical result of Eq.~\eqref{eq:MING_lossy_cat} is presented in Fig.~\ref{fig:MING_curves}(c). 

In order to find $d_F(\rho)$, note that $d_F(\rho)$ can be expanded as
\begin{align}
    d_F(\rho) &= \|U_{\pi/4}(\rho\otimes \rho) U^\dag_{\pi/4} - (\rho\boxplus\rho)\otimes(\rho\boxminus\rho)\|_F \\
    &= \left\{Tr{\rho^2}^2 + \Tr{(\rho\boxplus\rho)^2}\Tr{(\rho\boxminus\rho)^2} - 2\Tr{U_{\pi/4}(\rho\otimes \rho) U^\dag_{\pi/4}(\rho\boxplus\rho)\otimes(\rho\boxminus\rho)}\right\}^{1/2} \\
    &= \left\{\Tr{\rho^2}^2 + \Tr{(\rho\boxplus\rho)^2}^2 - 2\Tr{U_{\pi/4}(\rho\otimes \rho) U^\dag_{\pi/4}(\rho\boxplus\rho)\otimes(\rho\boxplus\rho)}\right\}^{1/2} \\
    &= \left\{\Tr{T_1^2}^2 + \Tr{T_2^2}^2 - \frac{2}{N_+^4}\Tr{\left(\Upsilon \otimes \Upsilon\right)^\dag U_{\pi/4}(\rho\otimes \rho) U^\dag_{\pi/4}\left(\Upsilon \otimes \Upsilon\right)(A\otimes A)}\right\}^{1/2}\;.
\end{align}
In the above expression, $T_1$, $T_2$, $\Upsilon$ and $A$ have already been derived or defined. The remaining task is finding the cross term under the square root, we expand
\begin{equation}
    \begin{aligned}
        U_{\pi/4}(\rho\otimes \rho)U_{\pi/4}^\dag &= \frac{1}{N_+^2}\sum_{s_1,s_2\in\{+,-\}} n_{s_1}n_{s_2} U_{\pi/4}(\phi_{s_1}\otimes \phi_{s_2})U_{\pi/4}^\dag\\
        &= \frac{1}{N_+^2}\sum_{s_1,s_2\in\{+,-\}} n_{s_1}n_{s_2} \proj{\omega_{s_1,s_2}}\;,
    \end{aligned}
\end{equation}
where $\ket{\omega_{s_1,s_2}}:=\frac{1}{2}\left[\ket{0}\otimes\ket{\sqrt{2}y} + s_1s_2 \ket{0}\otimes\ket{-\sqrt{2}y} + s_1\ket{\sqrt{2}y}\otimes\ket{0} + s_2\ket{-\sqrt{2}y}\otimes\ket{0}\right]$.
For convenience, let's define the following scalar-valued vector:
\begin{equation}
    \nu_{s_1,s_2}:=(\Upsilon\otimes \Upsilon)^\dag\ket{\omega_{s_1,s_2}} = 
    \frac{1}{2}
    \begin{bmatrix}
        e^{-y^2}(1+s_1)(1+s_2) \\
        (1+e^{-4y^2})(1+s_1s_2) + 2e^{-2y^2}(s_1+s_2) \\
        (1-e^{-4y^2})(1-s_1s_2) \\
        2e^{-2y^2}(1+s_1s_2) + (1+e^{-4y^2})(s_1+s_2) \\
        2(e^{-y^2}+e^{-5y^2})(1+s_1)(1+s_2) \\
        2(e^{-y^2}-e^{-5y^2})(1-s_1s_2) \\
        (1-e^{-4y^2})(s_1 - s_2)\\
        2(e^{-y^2} - e^{-5y^2})(s_1 - s_2)\\
        0
    \end{bmatrix}\;,
\end{equation}
where $\pm\equiv \pm 1$.
So the cross term is written as:
\begin{equation}
    -2\Tr{U_{\pi/4}(\rho\otimes \rho) U^\dag_{\pi/4}(\rho\boxplus\rho)\otimes(\rho\boxplus\rho)} = -\frac{2}{N_+^6}\sum_{s_1,s_2\in\{+,-\}} n_{s_1}n_{s_2} \nu_{s_1,s_2}^\dag(A\otimes A)\nu_{s_1,s_2}\;,
\end{equation}
which yields
\begin{equation}\label{eq:d_F_loss_cs}
    d_F(\rho) = \left[\Tr{T_1^2}^2 + \Tr{T_2^2}^2 -\frac{2}{N_+^6}\sum_{s_1,s_2\in\{+,-\}} n_{s_1}n_{s_2} \nu_{s_1,s_2}^\dag(A\otimes A)\nu_{s_1,s_2}\right]^{1/2}\;.
\end{equation}
The values of Eq.~\eqref{eq:d_F_loss_cs} with respect to different $z$ and $\gamma_L$ are presented in Fig.~\ref{fig:curves}(c). 

Note that as $\gamma_L \to 1$, any input state converges to the vacuum state, which is Gaussian. In this limit, we have $d_F = 0$ and $\MING_\alpha = 0$. In contrast, when $|z| \to \infty$ and $\gamma_L = 0$, the state $\rho \boxplus \rho = \psi \boxplus \psi$ becomes a density matrix with only two nonzero eigenvalues, each equal to $\frac{1}{2}$. This yields $\MING_1(\rho) = 2(S(\rho \boxplus \rho) - S(\rho)) = 2$.
In this extreme case, we also find $\Tr{\rho \boxplus \rho} \to \frac{1}{2}$ and
$\Tr{U_{\pi/4} (\rho \otimes \rho) U_{\pi/4}^\dagger  (\rho \boxplus \rho) \otimes (\rho \boxminus \rho)} \to \frac{1}{4}$,
so the Frobenius measure approaches $d_F \to \sqrt{1 + \frac{1}{2^2} - 2 \times \frac{1}{4}} = \sqrt{\frac{3}{4}}$.

When $|z| \to \infty$ and $0 < \gamma_L < 1$, the state $\rho$ asymptotically has two equal nonzero eigenvalues, each equal to $\frac{1}{2}$, giving $\Tr{\rho^2} \to 2 \times \left(\frac{1}{2}\right)^2 = \frac{1}{2}$. Meanwhile, $\rho \boxplus \rho$ asymptotically has three nonzero eigenvalues: $\frac{1}{2}$, $\frac{1}{4}$, and $\frac{1}{4}$, leading to $\Tr{(\rho \boxplus \rho)^2} \to \left(\frac{1}{2}\right)^2 + 2 \times \left(\frac{1}{4}\right)^2 = \frac{3}{8}$.
Additionally, $\Tr{U_{\pi/4}(\rho \otimes \rho) U_{\pi/4}^\dagger  (\rho \boxplus \rho) \otimes (\rho \boxminus \rho)} \to \frac{1}{8}$.
As a result, the Frobenius measure approaches $d_F \to \sqrt{ \left(\frac{1}{2}\right)^2 + \left(\frac{3}{8}\right)^2 - 2 \times \frac{1}{8} } = \frac{3}{8}$. These asymptotic behaviors are illustrated in Fig.~\ref{fig:curves}(c) and Fig.~\ref{fig:MING_curves}(a).

\end{appendix}

\end{document}